\documentclass[12pt]{article}
\usepackage{etex}

\title{Classifying grounded intersection graphs \\ via ordered forbidden patterns}
\author{Laurent Feuilloley$^{1,2}$ and Michel Habib$^2$ \\
$^1$ Univ. Lyon, Universit\'e Lyon 1, LIRIS, France \\
$^2$ IRIF, CNRS \& Paris University, France}

\usepackage[utf8]{inputenc} 
\usepackage[english]{babel} 
\usepackage{geometry} 
\makeatletter

\usepackage{subcaption} 
\usepackage{multicol}
\usepackage{comment}
\usepackage{amssymb} 
\usepackage{tikz} 
\usepackage{hyperref}

\usepackage{enumitem}
\usepackage{amsmath} 
\usepackage{amsthm} 

\theoremstyle{plain}
\newtheorem{Theo}{Theorem}

\theoremstyle{definition}
\newtheorem{Def}{Definition}
\newtheorem{lemma}{Lemma}
\theoremstyle{remark}
\newtheorem{Rk}{Remark}
\newtheorem{claim}{Claim}
\newtheorem{Obs}{Observation}

\newtheorem{open}{Open problem}

\newenvironment{claimproof}{\noindent\emph{Proof of the claim.}}{\hfill$\vartriangleleft$ \bigskip}

\newcommand{\CC}{\mathcal{C}}
\newcommand{\FF}{\mathcal{F}}

\input{rounding-rectangle-trick}

\definecolor{rouge1}{RGB}{170,  4,  0}
\definecolor{rouge2}{RGB}{218,  6,  0}
\definecolor{rouge3}{RGB}{255,  7,  0}
\definecolor{rouge4}{RGB}{255, 64, 59}
\definecolor{rouge5}{RGB}{255,148,146}
\definecolor{rouge6}{RGB}{255,170,170}

\definecolor{bleu1}{RGB}{14,  4,118}
\definecolor{bleu2}{RGB}{17,  3, 100}
\definecolor{bleu3}{RGB}{52, 51,250}
\definecolor{bleu4}{RGB}{93, 79,235}
\definecolor{bleu5}{RGB}{150,180,250}
\definecolor{bleu6}{RGB}{200,200,250}

\definecolor{bleu45}{RGB}{120, 120,240}

\definecolor{jaune1}{RGB}{170,152,  0}
\definecolor{jaune2}{RGB}{218,195,  0}
\definecolor{jaune3}{RGB}{255,228,  0}
\definecolor{jaune4}{RGB}{255,230, 20}
\definecolor{jaune5}{RGB}{255,243,146}

\definecolor{o6}{RGB}{255,237,158}
\definecolor{o5}{RGB}{255,231, 80}
\definecolor{o4}{RGB}{255,172, 50}
\definecolor{o3}{RGB}{255,120, 0}
\definecolor{o2}{RGB}{255,0, 0}
\definecolor{o1}{RGB}{150,0, 0}

\definecolor{g6}{RGB}{220,220,220}
\definecolor{g5}{RGB}{200,200,200}
\definecolor{g4}{RGB}{180,180,180}
\definecolor{g3}{RGB}{150,150,150}
\definecolor{g2}{RGB}{130,130,130}
\definecolor{g1}{RGB}{100,100,100}

\definecolor{gris_clair}{RGB}{220,220,220}
\definecolor{gris}{RGB}{180,180,180}
\definecolor{gris_fonce}{RGB}{100,100,100}
\definecolor{gg}{rgb}{0.35,0.35,0.35}
\definecolor{ggg}{rgb}{0.6,0.6,0.6}
\definecolor{gggg}{rgb}{0.7,0.7,0.7}

\definecolor{mon_orange}{RGB}{255,165,0}
\definecolor{mon_violet}{RGB}{200,90,200}
\definecolor{mon_violet2}{RGB}{190,150,220}

\definecolor{0-1}{RGB}{255,197, 57}
\definecolor{0-2}{RGB}{255,229,165}
\definecolor{0-3}{RGB}{255,211,104}
\definecolor{0-4}{RGB}{245,173,  0}
\definecolor{0-5}{RGB}{177,125,  0}

\definecolor{1-1}{RGB}{40,178,140}
\definecolor{1-2}{RGB}{146,225,203}
\definecolor{1-3}{RGB}{ 81,198,166}
\definecolor{1-3.5}{RGB}{40,180,136}
\definecolor{1-4}{RGB}{  0,161,116}
\definecolor{1-5}{RGB}{  0,116, 84}

\definecolor{2-1}{RGB}{255,128, 57}
\definecolor{2-1.8}{RGB}{255,210,185}
\definecolor{2-2}{RGB}{255,197,165}
\definecolor{2-2.5}{RGB}{255,170,125}
\definecolor{2-3}{RGB}{255,158,104}
\definecolor{2-3.5}{RGB}{255,133,52}
\definecolor{2-4}{RGB}{245, 88,  0}
\definecolor{2-5}{RGB}{177, 64,  0}

\definecolor{3-1}{RGB}{59, 81,184}
\definecolor{3-1.8}{RGB}{190,210,255}
\definecolor{3-2}{RGB}{157,169,227}
\definecolor{3-3}{RGB}{ 97,116,203}
\definecolor{3-3.5}{RGB}{55,80,180}
\definecolor{3-4}{RGB}{ 20, 46,167}
\definecolor{3-5}{RGB}{ 13, 32,121}
\definecolor{3-0}{RGB}{176,190,227}

\begin{document}

\maketitle{}

\begin{abstract}
It was noted already in the 90s that many classic graph classes, such as interval, chordal, and bipartite graphs, can be characterized by the existence of an ordering of the vertices avoiding some  ordered subgraphs, called \emph{patterns}. 
Very recently, all the classes corresponding to patterns on three vertices (including the ones mentioned above) have been listed, and proved to be efficiently recognizable. In contrast, very little is known about patterns on four vertices.

One of the few graph classes characterized by a pattern on four vertices is the class of intersection graphs of rectangles that are said to be \emph{grounded on a line}. 
This class appears naturally in the study of intersection graphs, and similar  grounded classes have recently attracted a lot of attention.

This paper contains three parts. First, we make a survey of grounded intersection graph classes, summarizing all the known inclusions between these various classes.
Second, we show that the correspondence  between a pattern on four vertices and grounded rectangle graphs is not an isolated phenomenon. We establish several other pattern characterizations for geometric classes, and show that the hierarchy of grounded intersection graph classes is tightly interleaved with the classes defined patterns on four vertices. We claim that forbidden patterns are a useful tool to classify grounded intersection graphs.
Finally, we give an overview of the 
complexity of the recognition of classes defined by forbidden patterns on four vertices and list several interesting open problems. 
\end{abstract}

\newpage{}


\section{Introduction}\label{sec:introduction}

There are several ways to approach well-structured graph classes, and two important ones are forbidden structures and intersection representations. 
In a nutshell, this paper is about establishing that there exists a strong relation between a type of forbidden structures and type of intersection representations. 

Some classic forbidden structures are  forbidden (induced) subgraphs and forbidden minors. 
In this paper, we will consider another type of forbidden structures, called \emph{patterns}. 
A pattern is basically an ordered graph, and we say that a graph avoids a pattern if there exists an ordering of its vertices such that the pattern does not appear~\cite{Damaschke90, FeuilloleyH21}.  
Various well-known classes have forbidden pattern characterizations, such as chordal, bipartite, comparability, permutation, split and threshold graphs. 
Characterizations by patterns are very useful. 
First, they can be more compact and manageable than forbidden induced subgraphs characterizations.
Second, they lead to efficient recognition algorithms, often based on graph traversals~\cite{HellMR14,FeuilloleyH21}.

Let us illustrate pattern characterization and their usefulness on the example of interval graphs.
A graph on~$n$~vertices is an interval graph if there exists a set of $n$~intervals that we can identify to the vertices, such that two intervals intersect if and only if the associated vertices are adjacent.  
It is well-known \cite{Olariu91, RamalingamR88} that interval graphs can also be characterized by a forbidden pattern. 
Namely, a graph is an interval graph, if and only if, there exists an ordering of its vertices, such that for any triplet of vertices, $a<b<c$, it is not the case that $(a,c)$ belongs to the graph, and $(a,b)$ does not. 
The forbidden pattern of this characterization is depicted in Figure~\ref{fig:pattern-drawing}. 

\begin{figure}[!h]
\begin{center}
\input{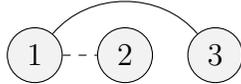}
\end{center}
\caption{\label{fig:pattern-drawing} Representation of the forbidden pattern for interval graphs. The general meaning of the edge styles is that in the forbidden pattern, the plain edges are present, the dashed edges are absent, and  if there is no edge between two vertices, there is no constraint.}
\end{figure}

This pattern characterization is very simple and compact, unlike the forbidden subgraph characterization of interval graphs, which consists in three graphs and two infinite families of graphs~\cite{LekkeikerkerB62}.
Also, thanks to this characterization, interval graphs can be recognized in linear time, because an ordering avoiding the pattern can be computed efficiently via graphs traversals \cite{CorneilOS09}.  

Interval graphs are also an example of intersection graphs.  
In a similar way as for interval graphs, given a set of geometric objects, one can define a graph by associating a node with every object and adding an edge between every two nodes whose geometric objects intersect. 
By allowing different types of geometric objects one can define different graph classes. 
Such geometric graphs have been well-studied, partly because they are related with practical algorithmic questions. 
The example of interval graphs illustrates how patterns and geometry can be related: the ordering of the left endpoints of the intervals, in the intersection representation, is an ordering that avoids the forbidden pattern. 

Another example of a known class that is characterized by both a pattern and an intersection model is the class of \emph{grounded rectangle graphs}. 
These graphs have been introduced a few years ago in several papers independently, under various names (diagonal-touched rectangle graphs~\cite{CorreaFPS15}, hook graphs~\cite{Hixon13}, max point-tolerance graphs~\cite{CatanzaroCFHHHS15}, and p-box(1) graphs~\cite{SotoC15}). 
Grounded rectangle graphs are the intersection graphs of rectangles grounded on a line, whose faces have angles $\pi/4$ and $3\pi/4$ with the grounding line. 
See Figure~\ref{fig:grounded-rectangles}. 
As proved independently in several of the papers cited above, this class is characterized by the forbidden pattern of Figure~\ref{fig:rectangle-patterns}. 
In other words, a graph is a grounded rectangle graph, if and only if, there exists an ordering of the nodes such that no subset of four nodes $a<b<c<d$, is such that: $(a,c)$ and $(b,d)$ are edges, and $(b,c)$ is not and edge.
Again, the ordering on the grounding line is an ordering that avoids the pattern. 
Note that the pattern of the interval graphs is included in the pattern of the grounded rectangle graphs. This is not a coincidence: grounded rectangle graphs generalize interval graphs, and this is witnessed by the pattern inclusion.

\begin{figure}[!h]
\centering
\begin{subfigure}{0.4 \textwidth}
	\centering
	\includegraphics[scale=1]%
	{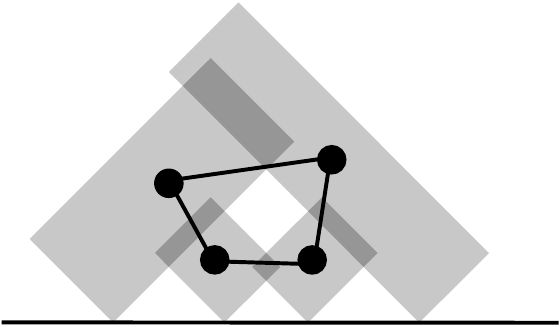}
	\vspace{0.2cm}
	\caption{\label{fig:grounded-rectangles} Grounded rectangles.}\end{subfigure}
\begin{subfigure}{0.4 \textwidth}
	\centering
	\vspace{1cm}
	\input{rectangle.tex}
	\vspace{1cm}
	\caption{\label{fig:rectangle-patterns} Associated pattern.}
\end{subfigure}
\caption{Grounded rectangle graphs and their forbidden pattern on 4 vertices.}
\end{figure}

In this paper, we investigate the relations between classes defined by intersecting objects grounded on a line, and classes defined by patterns of a specific shape. 
Beyond the fact that this connection is intriguing, we have two motivations for this study.

The first motivation comes from the study of patterns. 
A recent work \cite{FeuilloleyH21} established a complete characterization of the classes characterized by a set of forbidden patterns on three nodes. 
A remarkable fact is that basically all these classes have been studied before, and that they can all be recognized in polynomial time (and for most of them linear time). 
It is then natural to ask whether patterns on four nodes also define interesting classes that can be recognized efficiently. 
Answering this question seems very challenging. 
First, note that there are $3^6= 729$ different patterns on four vertices, thus a priori $2^{729}$ sets of patterns on 4 vertices. 
Even if many of these classes are equivalent because symmetries, it seems unrealistic to study all these classes exhaustively.
Second,  we know that the complexity landscape is going to be more complex, as there exists patterns on four vertices that define classes that are NP-hard to recognize.
Given this challenge, a reasonable place to start the study of patterns on four vertices is to generalize the geometric-pattern connection of grounded rectangle graphs.

The second motivation is on the geometry side. 
Establishing that two intersection graph classes are distinct can be a hard task, because it is not always easy to reason about geometric objects. 
In particular, it might be difficult to be exhaustive in a case analysis. 
On the other hand, considering an ordering of the nodes can be handy, because these are more combinatorial objects that can be easily enumerated. 
This also implies that it might be easier to find answers or conjectures by computer simulation.

\section{Overview}
\label{sec:overview}

Grounded rectangle graphs are just one example of grounded intersection graphs. By changing the shape, and the way the shapes are grounded, one can define and redefine many graph classes. 
For example, interval graphs correspond to shapes that are segments, and that are grounded on their two extremities. 
Other examples of well-known graph classes captured by this framework are permutation, cocomparability and interval filament graphs.
Many papers have defined new such classes, and established inclusions or incomparability between them. 
At this point, it is useful to have a common terminology (which we establish in Section~\ref{sec:definitions}), and a systematic survey of the classes and their inclusions (which we do in Section~\ref{sec:survey}).

In Sections~\ref{sec:touching} to \ref{sec:imcomparable}, we establish new equalities, inclusions and incomparabilities, between classes defined by forbidden patterns on four vertices and grounded intersection classes.
Let us give a simple example illustrating the type of results we get. 

Consider a configuration of unconstrained shapes in the plane grounded on a line, like in Figure~\ref{fig:grounded-intro}. 
We claim that the ordering of the shapes on the line avoids the pattern of Figure~\ref{fig:pattern-intro}. 

\begin{figure}[!h]
\centering
\begin{subfigure}{0.4 \textwidth}
\centering
\includegraphics[height=2 cm]%
{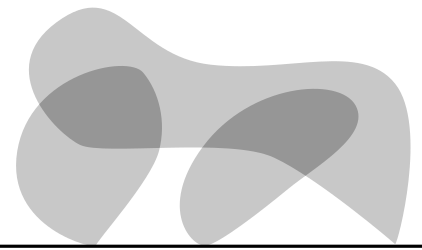}
\caption{\label{fig:grounded-intro}A grounded intersection model}
\end{subfigure}
\begin{subfigure}{0.4 \textwidth}
\centering
\input{grounded-pattern.tex}
\caption{\label{fig:pattern-intro} A pattern on four nodes}
\end{subfigure}
\caption{\label{fig:sample} A grounded intersection model and a related pattern on four nodes.}
\end{figure}

Let us sketch a proof (see Section~\ref{sec:convex-string}, for a proper proof). 
Suppose that the pattern appears when we consider the graph with the ordering of the shapes on the line. 
Consider the two shapes that correspond to position~1 and~3 in the pattern. 
They define two regions (above the line): the region below the shapes (that we could call the innerface) and the region above the shapes (the outerface). 
Since 2 is in the innerface, and 4 in the outerface, there should be an intersection between the shapes of numbers $\{1,3\}$ and $\{2,4\}$. which contradicts the fact that in the pattern of Figure~\ref{fig:pattern-intro}, no edge is allowed except $(1,3)$ and $(2,4)$.

Since we have not yet introduced all the notations, at the moment we cannot describe our results in details. Nevertheless, we will list them, as a roadmap for the paper that the reader can go back to, to get the big picture. 
Our results are enumerated below, and are illustrated (along with known results) in Figure~\ref{fig:diagram}. 
\begin{itemize}
\item Theorem~\ref{thm:touching-L-shapes} : The touching grounded L-shape graphs are the forests.
\item Theorem~\ref{thm:touching-strings}: The touching grounded string graphs are the outerplanar graphs.
\item Theorem~\ref{thm:bigrounded} (Edges 11 and 14 in Figure~\ref{fig:diagram}): Interval filament graphs (that we call bigrounded string graphs) are between $C_a$ and $C_{ab}$ in the inclusion hierarchy.
\item Theorem~\ref{thm:grounded-stairs} (Edge 19 and 22): Grounded stairs graph are between $C_{ab}$ and $C_{abc}$ in the inclusion hierarchy.
\item Theorem~\ref{thm:convex-Cabc} (Edge 24): Grounded convex graphs are included in $C_{abc}$.
\item Theorem~\ref{thm:grounded-strings} (Edge 27): String graphs (that we call grounded string graphs) are included in $C_{abcd}$.
\item Theorem~\ref{thm:pattern-strict} (Edge 13): $C_a$ is strictly included in $C_{ab}$, and $C_{abcd}$ does not contain all graphs.
\end{itemize}

The inclusions of grounded stairs and string graphs in, respectively, $C_{abc}$ and $C_{abcd}$ (Edges 22 and 24) can actually be derived from ~\cite{PachT19,PachT20}. 
We reprove them in our setting for completeness.

The rest of the inclusions described in Figure~\ref{fig:diagram} correspond to the following observations.

\begin{itemize}
\item Observation~\ref{obs:Cempty-Ca-Cb}: Edges 3, 9, 10.
\item Observation~\ref{obs:seg-stairs}: Edge 20.
\item Observation~\ref{obs:inclusion-shapes}: Edge 23 and 25.
\item Observation~\ref{obs:inclusion-patterns}: Edge 26.
\item Observation~\ref{obs:inequalities}: All the other edges.
\end{itemize} 

We can already make several observations about Figure~\ref{fig:diagram}.
We can see some classes we have mentioned: in the middle, grounded rectangle graphs and their pattern, and on the bottom interval graphs, and their pattern.
The inclusion we have just sketched is the edge labeled 27 (general shapes correspond to what we will call strings, and $C_{abcd}$ is the name we use for the class of the pattern of Figure~\ref{fig:pattern-intro}).
In general, we can see that pattern classes and geometric classes are intriguingly interleaved.

\begin{figure}[!h]
\begin{center}
\scalebox{0.85}{
\input{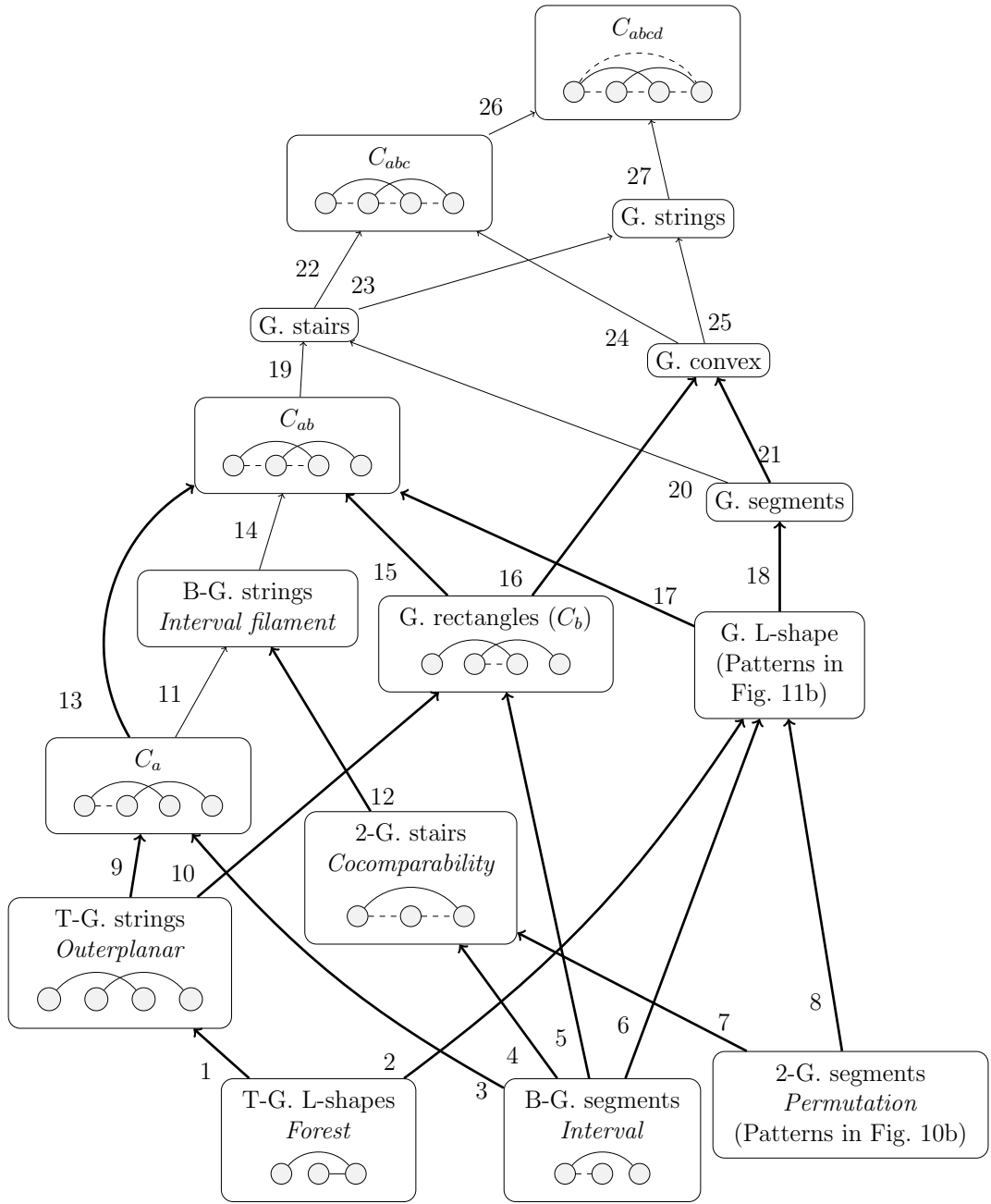}}
\end{center}
\caption{\label{fig:diagram} \textbf{Diagram of our results.}
The name of the classes in our vocabulary are written in normal font, the names in the literature are written in italic font.
Edges represent inclusions. Thick edges are known strict inclusions. 
G. stands for \emph{grounded}, T-G stands for \emph{touching grounded}, B-G. stands for Bigrounded, and 2-G. stands for 2-grounded.
}
\end{figure}

In the last section of the paper (Section~\ref{sec:recognition}), we review known results and open questions about the complexity of the recognition of the classes defined by forbidden patterns on four vertices.

\section{Preliminaries}
\label{sec:definitions}

\subsection{Patterns}
\label{subsec:preliminary-patterns}

We first define formally patterns, and the related concepts. We use the exact same vocabulary as in \cite{FeuilloleyH21}.

\begin{Def}
A \emph{trigraph} $T$ is a $4$-tuple $(V(T),E(T),N(T),U(T))$ where $V(T)$ is the vertex set and every unordered pair of vertices belongs to one of the three disjoint sets $E(T)$, $N(T)$, and $U(T)$ called respectively \emph{edges}, \emph{non-edges} and \emph{undecided edges}. A graph $G=(V(G),E(G))$ is a \emph{realization} of a trigraph $T$ if $V(G)=V(T)$ and $E(G)=E(T)\cup U'$, where $U'\subset U(T)$.
\end{Def}

\begin{Def}
An \emph{ordered graph} is a graph given with a total ordering of its vertices. We define the same for a trigraph, and call it a \emph{pattern}. 
We say that an ordered graph is a \emph{realization} of a pattern, the ordered graph and the pattern share the same set of vertices, the same linear order, and the graph is a realization of the trigraph.
\end{Def}

When drawing a pattern, we will use plain edges for edges, dashed edges for non-edges, and nothing for undecided edges (just like we did in the introduction). Also, in our drawings, the patterns are ordered from left to right.

\begin{Def}
In an ordered graph, if no subgraph is the realization of given pattern, we say that the ordered graph \emph{avoids} the pattern. 
Given a family of patterns $\mathcal{F}$, we define the class $\CC_{\FF}$ as the set of connected graphs that have an ordering avoiding the patterns in $\mathcal{F}$. 
\end{Def}

If $\mathcal{F}$ consists of only one pattern $P$, we abuse notation an simply write $\CC_{P}$, instead of $\CC_{\{P\}}$.
Note that one can equivalently consider forbidden ordered subgraphs instead of forbidden order trigraphs (which are the patterns). 
Like in \cite{FeuilloleyH21}, we prefer patterns, because they give more compact characterizations of the obstructions. 
Table~\ref{table:one-pattern-3} shows a list of graph classes, and their corresponding pattern.

\begin{table}[!h]
\begin{center}
\fbox{
\begin{tabular}{cc}
Graph Class & Forbidden Pattern  \\
\hline
\\[-0.15cm]
Linear forests & \input{path.tex} \vspace{0.27cm}\\
Stars & \input{star-3.tex} \\
Interval graphs & \input{interval-2.tex}  \vspace{0.27cm}\\
Split graphs& \input{split.tex}\\
Forest & \input{tree-2.tex} \vspace{0.27cm}\\
Bipartite graphs & \input{bipartite-3.tex}\\
Chordal graphs& \input{chordal-2.tex}\\
Comparability graphs & \input{comparability-2.tex}\\
Triangle-free graphs & \input{trianglefree-2.tex}\\[0.27cm]
At most two nodes & \input{null-2.tex} 
\end{tabular}
}
\vspace{0.5cm}
\caption{\label{table:one-pattern-3} The table from \cite{Damaschke90}. For each row the class on the first column is characterized by the pattern represented in the second column.}
\end{center}
\end{table}

As already noticed in \cite{Damaschke90, FeuilloleyH21, Wood06}, inclusions of patterns imply inclusions of classes.

\begin{Obs}
\label{obs:inclusion-patterns}
Consider two patterns $P$ and $P'$. 
We say that $P$ is included in $P'$ if $V(P)\subseteq V(P')$, $E(P)\subseteq E(P')$, and $N(P)\subseteq N(P')$. 
In this case, $\CC_P \subseteq \CC_{P'}$.  
\end{Obs}

There is one type of pattern that will be important in this paper, and that we define now.

\begin{Def}
\label{def:Pabcd}
Consider patterns on four nodes $\{1,2,3,4\}$. We use the following names for some pairs of nodes: $a=(1,2)$, $b=(2,3)$, $c=(3,4)$, $d=(1,4)$ (see Figure~\ref{fig:abcd-patterns}).
Let $S \subseteq \{a,b,c,d\}$, the pattern $P_{S}$ is the pattern on four vertices, with edges $(1,3)$ and $(2,4)$, and non-edge set $S$. The class $\CC_{S}$ is the class where $P_{S}$ is forbidden.
\end{Def}

\begin{figure}[!h]
\centering
\scalebox{1}{
\input{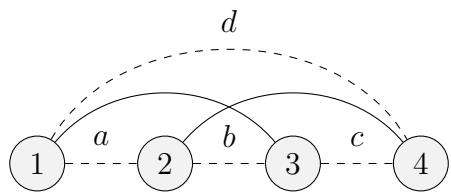}}
\caption{\label{fig:abcd-patterns}
Illustration of $P_{abcd}$ (Definition~\ref{def:Pabcd}).}
\end{figure}

\begin{Rk}
\label{Sym}
It is easy to see that a pattern and its mirror (the same pattern but with reversed order) define the same class.
Consequently $\CC_{a}=\CC_c$, $\CC_{a, d}=\CC_{c,d}$, $\CC_{a, b}=\CC_{b, c}$, $\CC_{b,c,d}=\CC_{a, b, d}$. 
Therefore, we can restrict attention to the following classes: $\CC_{\emptyset}$, $\CC_a$, $\CC_b$, $\CC_d$, $\CC_{a, b}$, $\CC_{a,c}$, $\CC_{a, d}$, $\CC_{b, d}$, $\CC_{a, b, c}$, $\CC_{a, b, d}$, $\CC_{a, c, d}$ and $\CC_{a, b, c, d}$.
\end{Rk}

Note that we have already mentioned that grounded rectangle graphs correspond to the pattern~$P_b$, and we discussed the pattern~$P_{abcd}$ at the beginning of Section~\ref{sec:overview}. 
It is folklore that outerplanar graphs are chacterized by $P_{\emptyset}$. 

\subsection{Grounded types}
\label{subsec:grounded-types}

\begin{Def}
We say that a configuration of shapes is grounded if it is in one of the two following situation. 
\begin{itemize}
\item There exists a horizontal straight line such that all the shapes are above the line, and are touching it in at least one point,
\item There exists a circle such that all the shapes are inside the circle, and are touching it in at least one point.
\end{itemize} 
We call the line (resp. circle) the \emph{grounding line} (resp. \emph{grounding circle}).
We also require that no two shapes touch the grounding line/circle on the same point (except for interval graphs, for which we require that the endpoints of the intervals are different). 
For the case of a grounding line, we can define the \emph{grounding ordering} as the order in which the shapes touch the line for the first time (from left to right). 
\end{Def}

We consider six grounding types, that are illustrated  Figure~\ref{fig:grounding-types}. 

\begin{Def}
\label{def:grounding-types}
A configuration of shapes is:
\begin{enumerate}[label=(\alph*), itemsep=-0.5ex]
\item \emph{Grounded}, if all the shapes are touching the \emph{grounding line} at least once.
\item \emph{Touching grounded}, if it is grounded, and the shapes only intersect on their borders, and not on their interiors. 
\item \emph{Outer}, if it has a grounding circle.
\item \emph{Bigrounded}, if every shape is touching the line on at least two points, and these two points are extremal in terms of abscissa.
\item \emph{2-Grounded}, if there is a second line above the grounding line, an all shapes are contained between the two lines, and touching the two lines.
\item \emph{Circle}, if it has a grounding circle that every shape touches twice.
\end{enumerate}
\end{Def}

\begin{Obs}
The touching grounded configuration is not well-defined for shapes that do not have interiors (\emph{e.g.} curves). 
For example, two segments that are crossing are intuitively not considered as touching, although they only intersect at their border.
In the following we will use a natural way to differentiate between configurations where curves are touching and where the curves are intersecting.
If, when we replace the lines by thick lines (with an interior), the drawing can be kept identical (that is only doing infinitesimal changes) then we will say that it is proper touching representation. 
If instead, replacing the lines by thick lines transforms the configuration into a drawing that cannot be fixed locally by minor modifications, then we will not consider the configuration as touching.  
\end{Obs}

\begin{figure}[!h]
\begin{center}
	\begin{subfigure}{0.3 \textwidth}
		\centering
		\includegraphics[height=2 cm]%
		{grounded-3.png}	
		{}%
		{\caption{\label{fig:grounded}
		Grounded.}}
	\end{subfigure}
	\hspace{0.5cm}
	\begin{subfigure}{0.3 \textwidth}
		\centering
		\includegraphics[height=2.5cm]
		{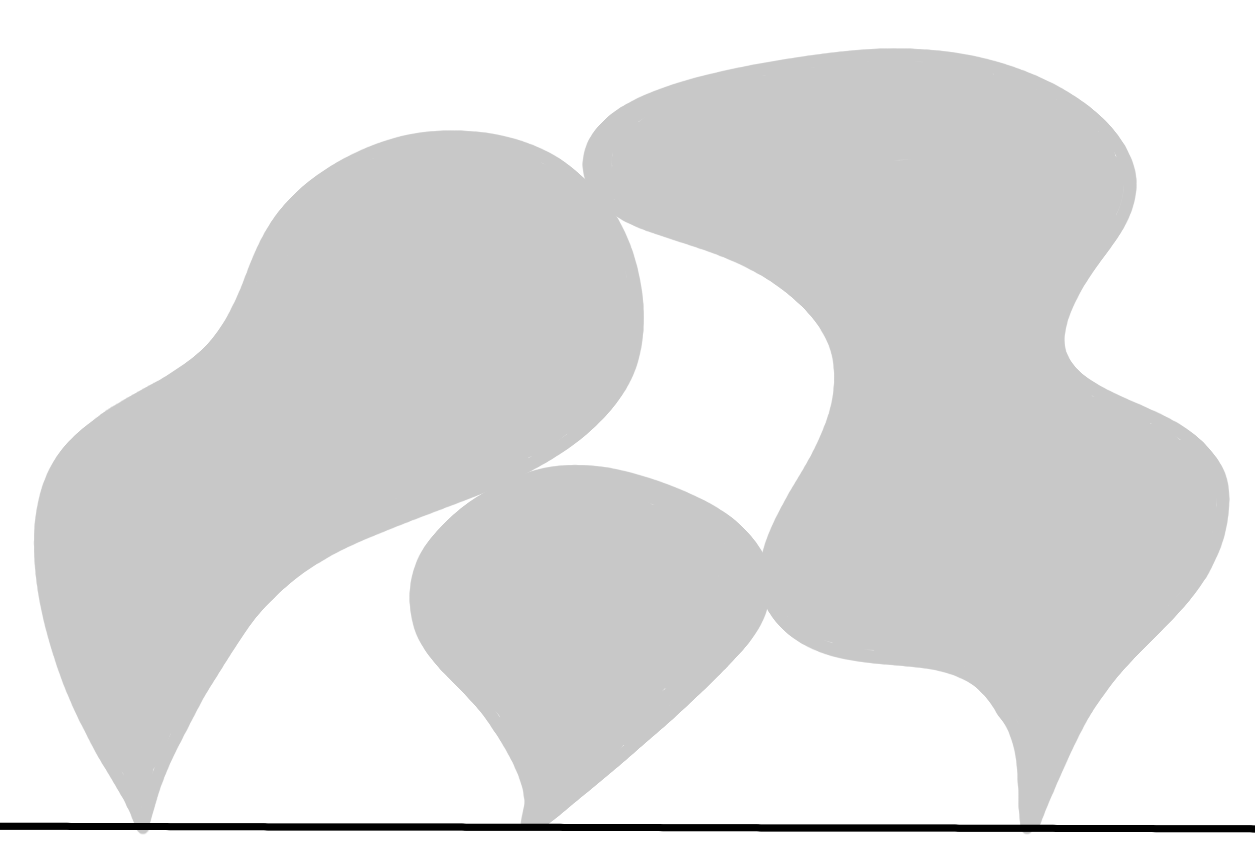}
		{}%
		{\caption{\label{fig:touching}
		Touching grounded.}}
	\end{subfigure}
	\hspace{0.5cm}
	\begin{subfigure}{0.3 \textwidth}
		\centering
		\includegraphics[height=2.5cm]
		{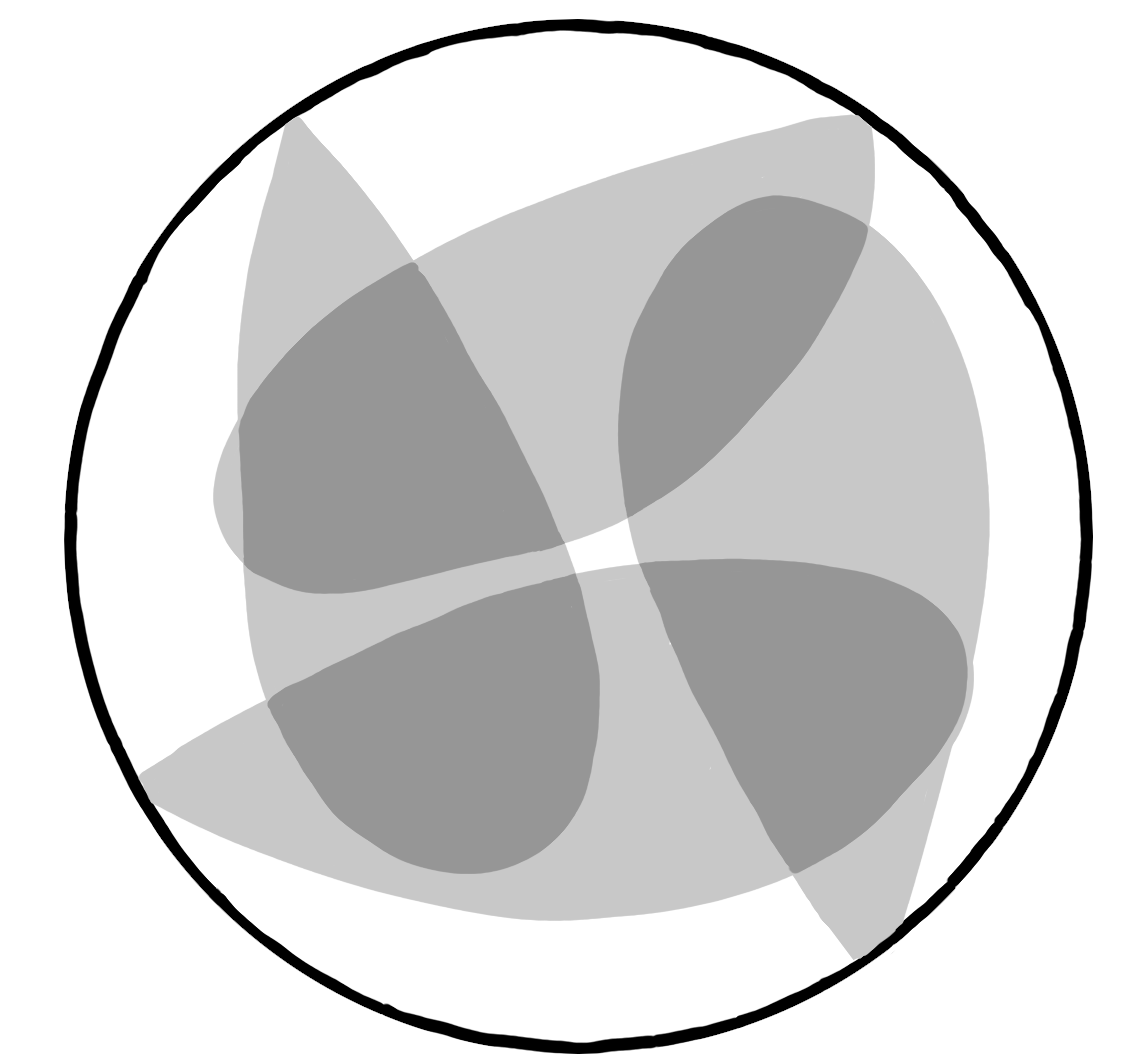}
		{}%
		{\caption{\label{fig:outer}
		Outer.}}
	\end{subfigure}
		\\[0.5cm]
		\begin{subfigure}{0.3 \textwidth}
		\centering	
		\includegraphics[height=2cm]
		{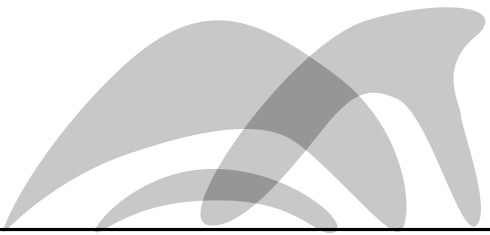}
		{}%
		{\caption{\label{fig:bigrounded}
		Bigrounded.}}
	\end{subfigure}	
	\hspace{0.5cm}
	\begin{subfigure}{0.3 \textwidth}
		\centering
		\includegraphics[height=2.5cm]
		{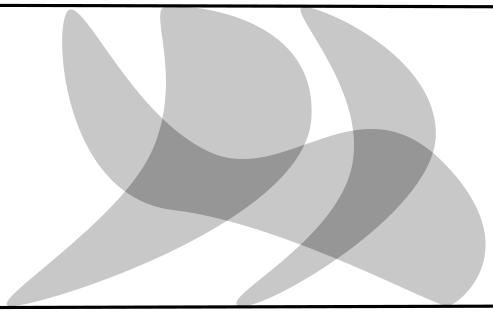}
		{}%
		{\caption{\label{fig:2-grounded}
		2-grounded.}}
	\end{subfigure}	
	\hspace{0.5cm}
	\begin{subfigure}{0.3 \textwidth}
		\centering
		\includegraphics[height= 2.5cm]
		{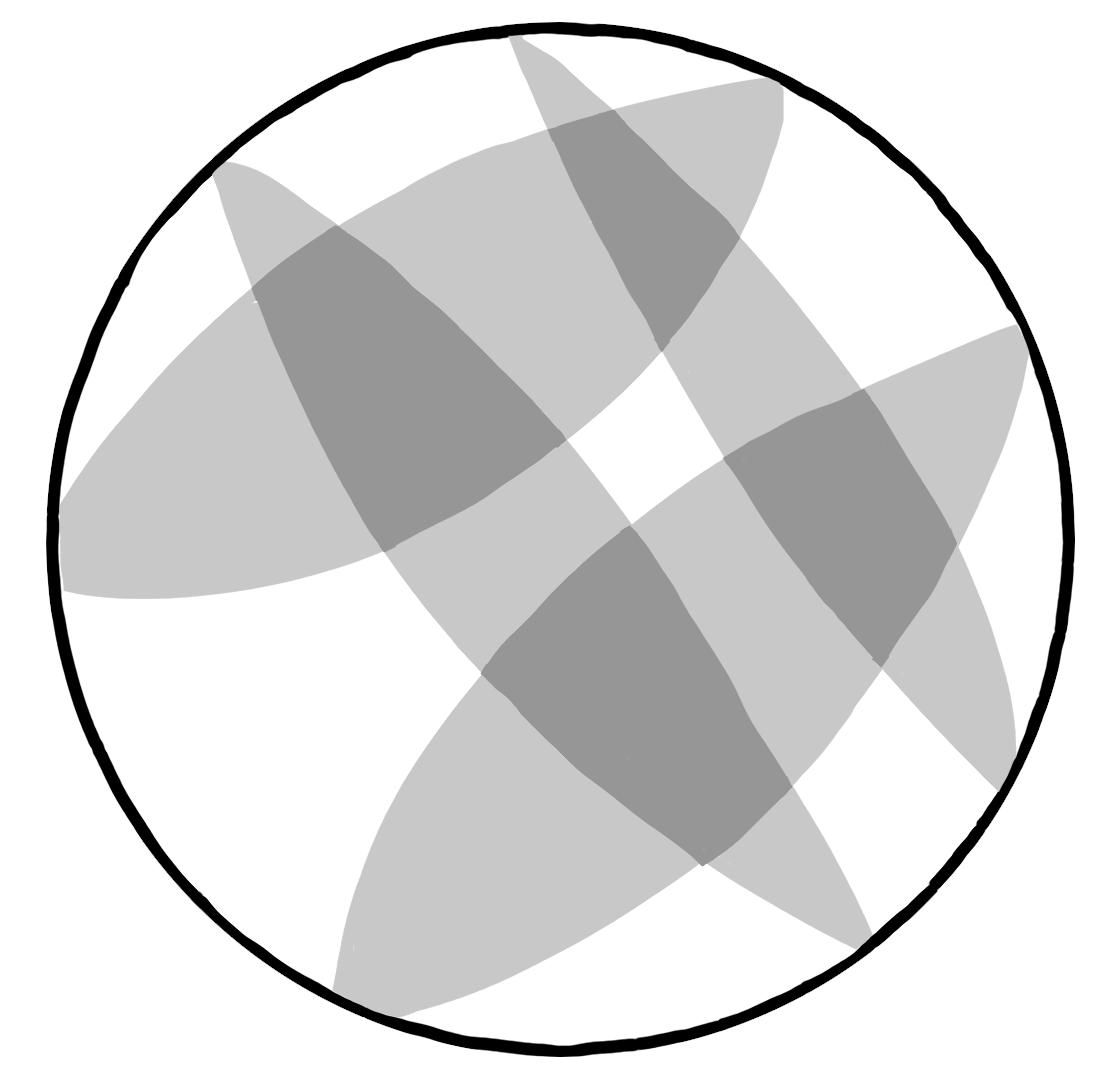}
		{}%
		{\caption{\label{fig:circle}
		Circle.}}
	\end{subfigure}
	\caption{\label{fig:grounding-types}  Grounding types of Definition~\ref{def:grounding-types}.}
\end{center} 
\end{figure}

\subsection{Relations between grounding types}
\label{subsec:groundings}

The different grounding types are related, and some inclusions follow automatically from these relations. Let us list these relations. 
We refer to a generic shape type $S$, that we are going to properly defined later, but one can think of $S$ as being segments, convex polygons, lines etc. 
We start by making a few informal observations, to show how one can transform one grounding type into another. Note that some shapes are more difficult to manipulate, and that some observation work only for some shapes.

\begin{Obs}\label{obs:inclusions-grounding}
For every shape $S$, touching-grounded-S graphs, bigrounded-S graphs and 2-grounded-S graphs are included in  grounded-S graphs, because the definitions of the former are more constrained than the definition of the later. 
For the same reason, circle-$S$ graphs are contained in outer-$S$ graphs.
\end{Obs}

\begin{Obs}\label{obs:line-to-circle}
For several shape types $S$, grounded-S graphs are included in outer-$S$ graphs, and bigrounded graphs are included into circle-$S$ graphs. 
This is because one can transform a straight line into a small arc of a large circle via a small deformation. If $S$ is stable by such deformation, then we can transfer the intersection representations. 
See Figure~\ref{fig:line-to-circle} for the case of grounded-S and outer-S configurations.
\end{Obs}

\begin{figure}[!h]
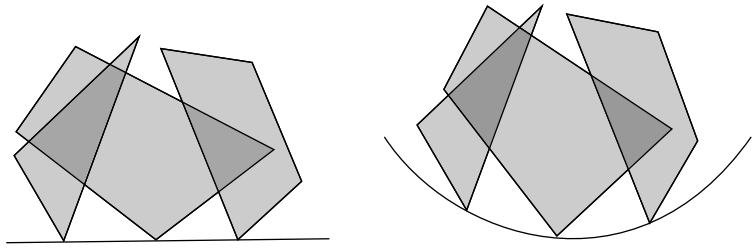

\begin{center}
\begin{tabular}{cc}
\begin{minipage}{0.3 \textwidth}
\input{line-to-circle-1.tex}
\vspace{0.4cm}
\end{minipage}
&
\begin{minipage}{0.3 \textwidth}
\input{line-to-circle-2.tex}
\end{minipage}
\end{tabular}
\vspace{-0.9cm}
\end{center}
\caption{\label{fig:line-to-circle} 
Illustration of Observation~\ref{obs:line-to-circle} with grounded-polygons and outer-polygons.}
\end{figure}

\begin{Obs}\label{obs:two-lines-to-circle}
For several shape types, 2-grounded-$S$ graphs are included into circle-$S$ graphs. 
This because two parallel lines can be transformed into a slice of a large circle with a small deformation. See Figure~\ref{fig:two-lines-to-circle}
\end{Obs}

\begin{figure}[!h]
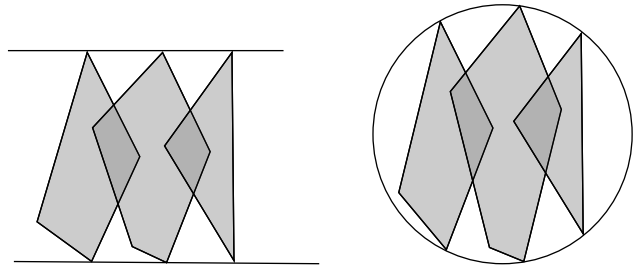

\begin{center}
\begin{tabular}{cc}
\input{two-lines-to-circle-1.tex}
&
\input{two-lines-to-circle-2.tex}
\end{tabular}
\end{center}
\caption{\label{fig:two-lines-to-circle}Illustration of Observation~\ref{obs:two-lines-to-circle} with 2-grounded-polygons and circle-polygons.}
\end{figure}

\begin{Obs}[\cite{Gavril00}]\label{obs:two-lines-one-line}
For several shape types, 2-grounded-$S$ graphs are included into bigrounded-$S$ graphs.
This is because we can take the upper line and plug it to the first line, continuously curving the shapes without creating new intersections. In this case there is a point $P$ on the line such that any shape has one grounding point on the left of $P$ and one grounding point on the right of $P$. See Figure~\ref{fig:two-lines-one-line}.
\end{Obs}

\begin{figure}[!h]
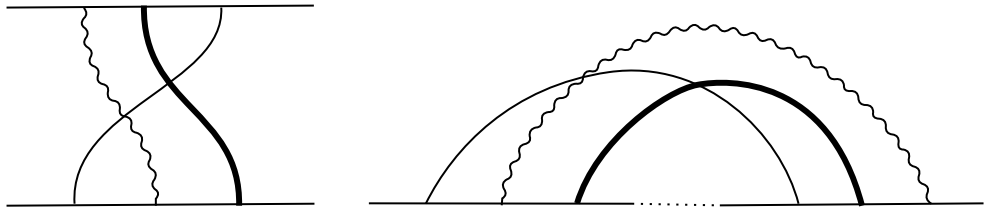

\begin{center}
\begin{tabular}{cc}
\input{two-lines-one-line-1.tex}
&
\input{two-lines-one-line-2.tex}
\end{tabular}
\end{center}
\caption{\label{fig:two-lines-one-line}
Illustration of Observation~\ref{obs:two-lines-one-line} with 2-grounded strings and bigrounded strings.}
\end{figure}

\begin{Obs}\label{obs:circle-to-line}
For several shape types, outer-$S$ graphs are included in grounded-S graphs.
This happens when one can cut the circle in one point, and then continuously transform both the circle and the shapes,  keeping the same intersections, until the circle is straighten into a line (in a  similar way as in the transformation of Observation~\ref{obs:two-lines-one-line})s.
\end{Obs}

\subsection{Shape types}
\label{subsec:shape-types}

We now list the shapes we are going to use. 

\begin{Def}
\label{def:shape-types}
A shape type is one the following 8 cases.
\begin{enumerate}[itemsep=-0.5ex]
\item \label{i:seg} \emph{Segments}.
\item \label{i:tri} \emph{Triangles}.
\item \label{i:trap} \emph{Trapeziums}.
\item \label{i:rect} \emph{Diagonal rectangles}, that are rectangles touching a grounding line on one corner, and with faces with angles $\pi/4$ and $3\pi/4$ with respect to the grounding line. 
\item \label{i:L} \emph{L-shape} that are formed by a vertical segment touching the grounding line, and extended by an horizontal segment to the right.
\item \label{i:stairs} \emph{Stairs} which are series of horizontal and vertical segments forming stairs going to the right. 
\item \label{i:poly} \emph{(Convex) polygons}. (For this paper all polygons are convex, thus we will simply say ``polygons''.)
\item \label{i:strings} \emph{Strings} which are simply curves.
\end{enumerate}
\end{Def}

\begin{Obs}\label{obs:inclusion-shapes}
Strings are more general than all the other shapes. Convex polygons are more general than segments and rectangles.
\end{Obs}


\section{Survey of grounded intersection graphs}\label{sec:survey}

In this section, we make a survey of grounded intersection graphs. The goal is to have a more detailed survey than in previous papers~\cite{JelinekT19, DaviesKMW21}, with a special focus on pattern characterization. It can be seen as an extended related works section.

We study the combinations of the shape types (Definition~\ref{def:shape-types}) and grounding types (Definition~\ref{def:grounding-types}).   
Some combinations are not mentioned, because they have not been studied, and are not used in this paper either.
This section contains eight subsections, one for each shape type of Definition~\ref{def:shape-types}. Each subsection contains a list of classes (with some additional remarks, such as pattern characterization) and a list of inclusions with other classes of the survey. 

\subsection{Segments}
\label{subsec:segments}

\paragraph{Classes}
Grounded segment graphs have been defined in \cite{CardinalFMTV18}. 
It is proved in \cite{CardinalFMTV18} that this class is the same as the one of \emph{downward rays intersections}, where a downward ray is a halfline pointing downward. 
 
The 2-grounded segment graphs are known to be the same as the \emph{permutation graphs}. 
A permutation graph is a graph that can be built from a permutation the following way: the vertices represent the elements of a permutation, and the edges represent pairs of elements that are reversed by the permutation.
Permutation graphs are characterized by a couple of forbidden patterns \cite{EvenPL72}, see Figure~\ref{fig:permutations}.

\begin{figure}[!h]	
	\begin{center}
	\begin{subfigure}[b]{0.35 \textwidth}
		\centering
		\includegraphics[scale=1]%
		{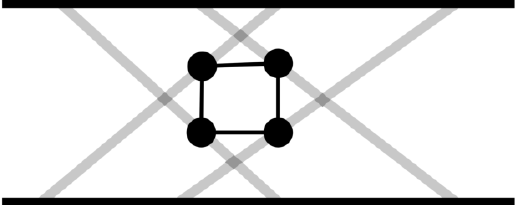}
		{}
		{\caption{
			Intersection graph}}
	\end{subfigure}
		\begin{subfigure}[b]{0.35 \textwidth}
			\centering
			\begin{tabular}{c}
			\scalebox{1.0}{
			\input{
			comparability.tex}}\\
			\scalebox{1.0}{
			\input{
			cocomparability.tex}}
			\end{tabular}
			\vspace{0.2cm}
			\caption{\label{fig:patterns-permutation}
			Forbidden patterns.}
		\end{subfigure}
	\caption{\label{fig:permutations}
	Two characterizations of permutation graphs.} 
			\end{center}
	\end{figure}
 
The bigrounded segment graphs are the \emph{interval graphs}. 
Indeed, if the segments touch the grounding line by two points, then they are completely contained in the grounding line. As already mentioned in the introduction, interval graphs have pattern characterization, with the pattern of  Figure~\ref{fig:pattern-drawing}.

The circle segment graphs are known as \emph{circle graphs}. It is known that outerplanar graphs are circle graphs~\cite{Unger88}.

\emph{Outer segment graphs} are introduced in \cite{CardinalFMTV18}, where it is proved that this class strictly contains the intersection graphs of rays (that is, of half-lines).

\paragraph{Inclusions.}

\begin{itemize}[itemsep=-0.5ex]
\item Grounded segments graphs contain 2-grounded segment graphs (permutation graphs) and  bigrounded segment graphs (interval graphs), by Observation~\ref{obs:inclusions-grounding}. These inclusions are strict (see \emph{e.g.}~\cite{JelinekT19}).
\item Grounded segments strictly contains grounded L-shapes \cite{JelinekT19}.
\item Outer-segment graphs contain circle segment graphs, by Observation~\ref{obs:inclusions-grounding}.
\item Grounded segment graphs are included into outer-segment graphs, by Observation~\ref{obs:line-to-circle}, and this inclusion is strict~\cite{CardinalFMTV18}.
\item 2-grounded segment graphs are included in circle segment graphs, by Observation~\ref{obs:two-lines-to-circle}, and the inclusion is strict (as witnessed by $C_6$).
\item By Observation~\ref{obs:inclusion-patterns}, interval graphs are included in outerplanar graphs.
\end{itemize}

\subsection{Triangles}

\paragraph{Classes.}
For triangles, we only consider the 2-grounded case, which actually contains two subcases.
The \emph{point interval graphs} (PI-graphs) are the intersection graphs of 2-grounded triangles, where every triangle has an edge on the bottom line.  
The PI$^*$-graph class is the generalization where every triangle has an edge on one of the two lines. 
Interestingly, PI-graphs can be recognized in polynomial-time~\cite{Mertzios15}, but PI$^*$-graphs are NP-complete to recognize~\cite{Mertzios12}.

\paragraph{Inclusions.}

PI-graphs contain bigrounded segment (interval) and 2-grounded segment (permutation) graphs, by generalization of the shapes, and the inclusions are strict (see~\cite{Mertzios12}).

\subsection{Trapeziums}

\paragraph{Classes}
The 2-grounded trapezium graphs, where the trapeziums have the two parallel faces on the two lines, are known as \emph{trapezoid graphs}  \cite{DaganGP88}. 
 
Circle-grounded trapezium graphs can be defined as intersections of regions of the circle defined by two non-crossing cords. These are known as \emph{circle trapezoid graphs}~\cite{FelsnerMW97}.
 
\paragraph{Inclusions}
The 2-grounded trapezium graphs are included in the circle trapezium graphs, via Observation~\ref{obs:two-lines-to-circle}, and the inclusion is strict~\cite{FelsnerMW97}.
The 2-grounded trapezium graphs contain PI$^*$ graphs by the preliminary observations, and the inclusion is strict~\cite{Corneil97}.

\subsection{(Diagonal) rectangles}
\label{subsec:rectangles}

\paragraph{Classes.}

Grounded diagonal rectangles graphs are known under various names as already mentioned in the introduction: diagonal-touched rectangle graphs \cite{CorreaFPS15}, hook graphs \cite{Hixon13}, max point-tolerance graphs \cite{CatanzaroCFHHHS15}, and p-box(1) graphs \cite{SotoC15}. 
In particular this class can be also obtained as intersection of diagonal orthogonal triangles or even diagonal L-shapes.
This class is characterized by the pattern $P_{b}$ (see Figure~\ref{fig:rectangle-patterns}).

\paragraph{Inclusions.}

Grounded diagonal rectangles graphs strictly contain interval graphs and outerplanar graphs~\cite{CatanzaroCFHHHS15}. 

\subsection{L-shapes}
\label{subsec:L}

\paragraph{Class.}

Grounded L-shapes were introduced in \cite{McGuinness96}. It was proved in \cite{JelinekT19} that these graph have a pattern characterization with the two patterns of Figure~\ref{fig:L-shapes}. 
The first of the two patterns says that on a set of nodes $1<2<3<4$, it is not possible to have edges $(1,3)$ and $(2,4)$, but not $(1,2)$ and $(2,3)$ (this is what we call $P_{ab}$). The second pattern says that it is not possible to have edges $(1,2)$, $(2,3)$ and $(1,4)$ but not $(1,3)$. 

\begin{figure}[h!]
	\begin{subfigure}{0.5 \textwidth}
			\centering
				\includegraphics[scale=1]				{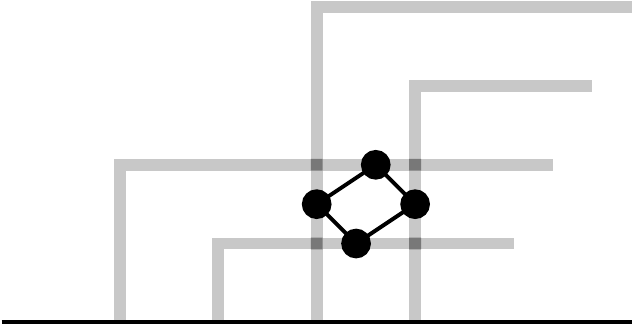}
			\vspace{0.2cm}
			\caption{
			\label{fig:L-shapes-cycle}
			Grounded L-shapes and the associated grounded-L graph.}
		\end{subfigure}
	\begin{subfigure}{0.5 \textwidth}
			\centering
				\begin{tabular}{c}
				\input{
				L-grounded-1.tex}\\
				\input{
				L-grounded-2.tex}
				\end{tabular}
			\vspace{0.2cm}
			\caption{\label{fig:patterns-L-shapes}
			Forbidden patterns}
		\end{subfigure}
	\caption{\label{fig:L-shapes}
	Definition and pattern characterizations of grounded L-graphs.}
\end{figure}

\paragraph{Inclusions.}
\begin{itemize}
\item Because the L-shapes can only be grounded on the lowest point, the outer-L-shape graphs are the same as the grounded-L-shape. 
\item The circle-L-shape graphs (where every L touches the circle by it lowest and right-most points), are the same as the circle-segment graphs (or simply circle graphs) \cite{JelinekT19}.
\item 2-grounded segments graphs, that is interval graphs, are included in grounded L-shape: simply use L's that all have the horizontal segments at the same altitude. 
\item  Both grounded L-shape and grounded rectangles are included in in $C_{ab}$ by Observation~\ref{obs:inclusion-patterns}. As these two classes are not comparable~\cite{JelinekT19}, these inclusions are strict.
\end{itemize}

\subsection{Stairs}
\label{subsec:stairs}

\paragraph{Class.}
For stairs, only the 2-grounded case has been considered.
The 2-grounded stairs graphs are exactly the cocomparability graphs. 
Indeed, it is known that cocomparability graphs are intersection graphs of continuous functions between two lines (see~\cite{GolumbicRU83} or the explanation we give in Subsection~\ref{subsec:strings}), and we can make these functions strictly increasing, and them discretize them to get stairs.
Comparability graphs are characterized by the pattern of Figure~\ref{fig:cocomp}. 

\begin{figure}[!h]	
		\begin{subfigure}[b]{0.6 \textwidth}
			\centering
			
				\includegraphics[scale=1.0]%
				{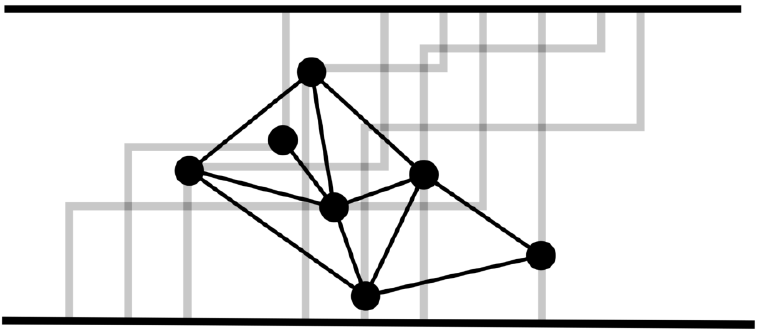}
			{}	
			\vspace{0.2cm}
			\caption{
			Intersection graph.}
		\end{subfigure}
		\begin{subfigure}[b]{0.4 \textwidth}
			\centering
			\scalebox{1.0}{
			\input{
			cocomparability.tex}}
			\vspace{0.2cm}
			\caption{ 
			Forbidden patterns.}
		\end{subfigure}
	\caption{\label{fig:cocomp}Two characterizations of cocomparability graphs.}
	\end{figure}

\paragraph{Inclusions}
In several settings stairs generalize segments. First suppose that the segments are pointing upward to the right. Then one can transform those into stairs: just discretize the segments into stairs, with a fine enough discretization to keep the same intersections. 
As a matter of fact, we can often assume that the segments are pointing to the right: if it is not the case, we can tilt all the segments to the right, until they are all pointing to the right.

This techniques works for example for 2-grounded stairs, that are included into 2-grounded segments (which is not a surprise as these classes boil down to permutation and cocomparability graphs).
It also allows to prove the following observation, that we use in Figure~\ref{fig:diagram}. 

\begin{Obs}
\label{obs:seg-stairs}
Grounded segment graphs are contained in grounded stairs graphs. 
\end{Obs}

\subsection{Convex Polygons}
\label{subsec:polygons}

Circle-polygon graphs were studied in \cite{KostochkaK97}. Grounded convex polygons contain grounded segments and grounded rectangles by Observation~\ref{obs:inclusion-shapes}.

\subsection{Strings}
\label{subsec:strings}

\paragraph{Classes}
Grounded string graphs are the same as outer-string graphs (by Observations~\ref{obs:circle-to-line} and~\ref{obs:line-to-circle}, as noticed in \cite{CardinalFMTV18}). Outerstring graphs are well studied. See Figure~\ref{fig:outer-string_inter} for a grounded representation of such a graph.
 
The 2-grounded string graphs are co-comparability graphs~\cite{GolumbicRU83}.
Let us give some detail on this topic.
Any shape between 2 lines separates the vertices (or shapes) in between two lines as left part and right part (shapes totally to the left or totally to the right). Therefore the relation : being totally to the left (resp. the right) is a partial order  and therefore its complement is a comparability graph.
This implies that every 2-grounded class  is included in cocomparablity graphs, and for 2-grounded string graphs this is an equality.

Bigrounded string graphs are known as \emph{interval filament graphs}~\cite{Gavril00}. See Figure~\ref{fig:filament_inter} for a representation.

\begin{figure}[h!]
\begin{center}
	\begin{subfigure}[b]{0.45 \textwidth}
			\centering
			\includegraphics%
			[width=\textwidth]%
			{
			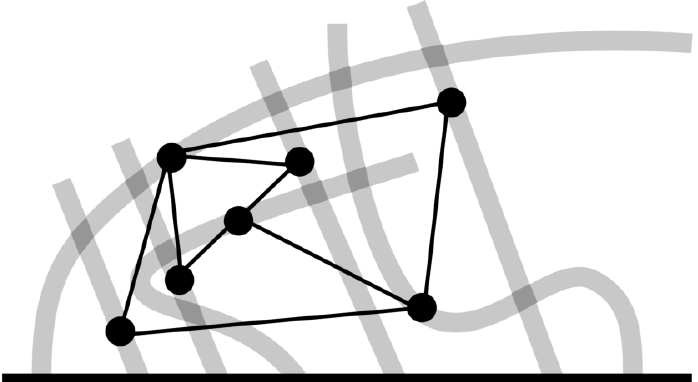} 
			{}
		\caption{\label{fig:outer-string_inter} An outer-string representation and graph.}
		\end{subfigure}
		\hspace{0.5cm}
		\begin{subfigure}[b]{0.45 \textwidth}
			\centering
		\includegraphics%
		[width=\textwidth]%
		{
				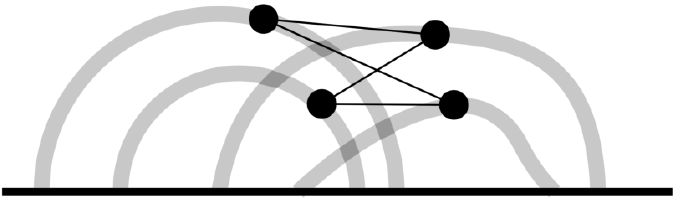} 
			{}
		\caption{\label{fig:filament_inter} 
		An interval filament representation and graph.}
		\end{subfigure}	
		\caption{Representations of an outer-string graph and an interval filament graph.}
\end{center}	
\end{figure}		

\paragraph{Inclusions}
\begin{itemize}
\item It was shown in \cite{Gavril00} that bigrounded string graphs contain 2-grounded string graphs (by what we called Observation~\ref{obs:two-lines-one-line}), and circle-polygon graphs, and that these inclusions are strict.
\item Outer-segment are strictly included in outer-string \cite{CabelloJ17}.
\item Strings generalize all the shapes we use. This is clear for segments, stairs and L-shape which are special strings, but it is also true for rectangles and convex polygons, as their boundaries are strings, and only the boundary matters for intersections.
\end{itemize}


\section{Touching grounded intersection characterization}
\label{sec:touching}

\begin{Theo}
\label{thm:touching-L-shapes}
The touching grounded L-shape graphs are the forests.
\end{Theo} 

Note that forests are also characterized by a pattern on three vertices (See Table~\ref{table:one-pattern-3}).

\begin{proof}

Before we prove both inclusions, let us
introduce a notation and make an observation.
First, a grounded L-shape $S$ can be characterized by three numbers: 
\begin{enumerate}
\item The abscissa of its rooting point, that we denote $x(S)$.
\item The abscissa of its right-most point, that we denote $x'(S)$.
\item The ordinate of it horizontal segment, that we denote $y(S)$.
\end{enumerate}

 When $S$ corresponds to a vertex $v$ in a graph, we will abuse notation, and write $x(v)$, $x'(v)$ and $y(v)$.  
Now, in a touching L-shape representation, two L-shapes $S$ and~$S'$, with $x(S)<x(S')$ can touch in only one way: $S$ touches the vertical part of of $S'$, with its right-most point.

Remember that forests are characterized by the pattern on three nodes with edges $(1,2)$ and $(1,3)$ (see Table~\ref{table:one-pattern-3}).
Now, consider the grounding order of an arbitrary set of touching grounded L-shapes. Take three shapes $S_1$, $S_2$ and $S_3$ in that order.
If $S_1$ and $S_2$ are touching, then this has to be on the right-most point of $S_1$, and on the vertical part of $S_2$. Then $S_1$ cannot go further right than the grounding position of $S_2$, because then it would cross $S_2$. 
(Remember that for string-like shapes, we defined ``touching'' as ``touching even if we were to make the strings have non-zero thickness'', thus crossings are forbidden for touching L-shapes.)
Consequently, $S_1$ cannot reach~$S_3$. 
Therefore, the grounding ordering of the shapes avoids the forbidden pattern of forests, and touching grounded L-shape graphs are included in forests. 

For the other direction, we build a touching L-shape representation for any forest. 
First, note that we can assume without loss of generality that the forest is actually a tree $T$, because the representations of the different connected component can be placed side by side on the line. 

Our construction is top-down, that is, we place the L-shape of the root first, and we go down in the tree towards the leaves. 
Consider the tree $T'$ of the vertices that have been added so far to the drawing. 
At first $T'$ is empty, and at the end $T'=T$.
Our construction will take place in the unit square $[0,1]\times [0,1]$. 
We first represent the root~$r$ as a vertical segment between positions $(1,0)$ and $(1,1)$. 
At that point $T'=\{r\}$. 
We explain how we can take a leaf from $T'$, and add all its children in $T$ to~$T'$, by adding them to the drawing. 
By induction, this describes a touching L-shape representation of the tree.

We maintain the following invariant. 
For any vertex $v$ that is a leaf in $T'$, there exists an abscissa $x^*(v)<x(v)$, and an ordinate $y^*(v)\leq y(v)$, such that the rectangle $(x^*(v),x(v)) \times [0,y^*(v)]$ intersects none of the shapes placed so far.
Note that this invariant is satisfied when we have just the root in $T'$, with $x^*(v)=0$ and $y^*(v)=1$.
Suppose, the invariant holds for an arbitrary leaf $v$ of $T'$. 
Now, consider the children of $v$ in $T$: $v_1,\dots,v_k$ (in an arbitrary order). 
For each $v_i$, we build the L-shape with the following characteristics (see Fig.~\ref{fig:touching-L}):  \[
x(v_i)=x^*(v)+i\cdot \frac{x(v)-x^*(v)}{k+1}\hspace{1cm}
x'(v_i)=x(v)\hspace{1cm}
y(v_i)= (k-i+1)\cdot \frac{y^*(v)}{k+1}.
\]

\begin{figure}[!h]
\begin{center}
\scalebox{1.2}{
\input{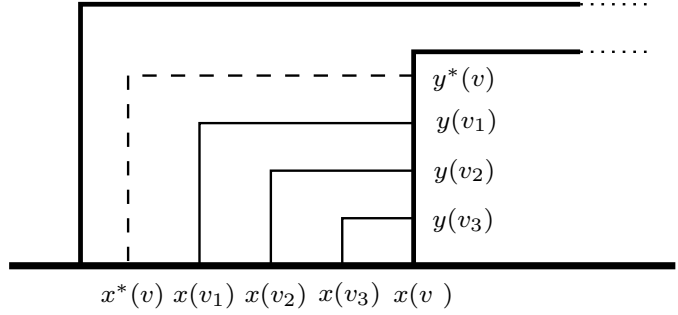}}
\caption{\label{fig:touching-L} Illustration of the construction of the proof of Theorem~\ref{thm:touching-L-shapes}.}.
\end{center}
\end{figure}

One can check that these shapes do not intersect, and that they all touch the shape of $v$. Also, because they are in the rectangle of our invariant, they touch no other shape. Finally, our invariant is preserved, since for any new leaf $v_i$, the rectangle with the following coordinates is empty: $x^*(v_1)=x^*(v)$ and $x^*(v_i)=x(v_{i-1})$ for $i>0$, and $y^*(v_i)=y(v_i)$, for all $i$.
\end{proof}
 
\begin{Theo}
\label{thm:touching-strings}
The touching grounded strings graphs are the outerplanar graphs.
\end{Theo} 
 
As already mentioned outerplanar graphs are also characterized by the pattern~$P_{\emptyset}$.

\begin{proof}
Consider an outerplanar graph, and let us build a representation by touching rectangles (which are a special case of touching strings). 
For a point $p$ of the grounding line, let its \emph{left diagonal} be the half-infinite line that starts at $p$, and goes in the north-west direction. Similarly define the \emph{right diagonal} for the north-east direction. 
Consider an ordering of the vertices of the graph that avoids the pattern $P_{\emptyset}$. 
Plant the corresponding rectangles regularly on the line in this ordering. 
For a given vertex~$v$, let $v_{\ell}$ and $v_r$ its left-most and right-most neighbors in the ordering (considering that $v$ is a neighbor to itself). 
Now the rectangle of $v$ is defined in the following way: 
\begin{itemize}
\item Its left-most corner is at the intersection of its left diagonal and of the right diagonal of $v_{\ell}$.
\item Its right-most corner is at the intersection of its right diagonal and of the left diagonal of $v_r$.
\end{itemize}

Clearly, this intersection graph will have all the intersection of the original graph. 
Let us prove that there are no other intersection. 
Suppose there is an intersection between the rectangles of vertices $a$ and $b$, with $a<b$, but $a$ and $b$ are not neighbors in the graph. 
Then, by construction it must be that $a$ has a neighbor $c$ such that $b<c$, and $b$ has a neighbor $d$ such that $d<a$. 
Consequently, the nodes $d<a<b<c$ form $P_{\emptyset}$, the forbidden pattern of outerplanar graphs, which is a contradiction. 

In the other direction: suppose we have a touching string representation, let us prove that the ordering of the shapes along the line corresponds to an ordering of the vertices that avoids the pattern. 
For the sake of contradiction, suppose we have four strings $a$, $b$, $c$ and $d$, such that $a$ and $c$ are in contact, and $b$ and $d$ are in contact. 
Strings $a$ and $c$ must touch at some position~$p$. 
Now the grounding point of $b$ is in the region delimited by the grounding line and the parts of $a$ and $c$ that starts on the grounding line and end in~$p$. 
The grounding point of $d$ is outside this region, and none of $b$ or $d$ can cross the boundaries of the region, thus $b$ and $d$ cannot meet, a contradiction.
\end{proof}

We can derive from the proof that outerplanar graphs also correspond to  touching grounded rectangle graphs, and touching grounded polygon graphs. 

\section{Interval filament graphs are between $C_a$ and $C_{ab}$}

In this section, we prove that the geometric class of interval filament graphs is between the pattern classes $C_a$ and $C_{ab}$.

\begin{Theo}\label{thm:bigrounded}
The following holds:
 $C_a \subseteq$ Interval filaments $\subseteq C_{ab}$.
\end{Theo} 

We think that it is insightful to have both inclusions into one theorem, but for convenience, we split the theorem into two lemmas for the proofs.

\begin{lemma}
Interval filaments $\subseteq C_{ab}$.\end{lemma}

\begin{proof}

Consider an interval filament representation. 
We claim the order of the left grounding points of the interval filaments avoids the pattern~$P_{ab}$.
For the sake of contradiction, suppose that the pattern appears in this ordering, on some vertices $a<b<c<d$. 
Let $\alpha_a,\alpha_b,\alpha_c,\alpha_d$ be the left endpoints, and $\beta_{a},\beta_{b},\beta_{c},\beta_{d}$ be the right endpoints of the filaments corresponding to $a$, $b$, $c$ and $d$. 
Because $b$ and $d$ are adjacent, it must be that $\alpha_d<\beta_b$. 
Similarly, because $a$ and $c$ are adjacent, $\alpha_c<\beta_a$. 
Since $a$ and $b$ are not adjacent, we cannot have $\alpha_a<\alpha_b<\beta_a<\beta_b$, therefore, given the inclusions above, we must have $\alpha_a < \alpha_b < \alpha_c <  \alpha_d <\beta_b<\beta_a$. 
Consequently, $\alpha_c$ is in the region of the plane delimited by the grounding line and the filament of $b$. 
As $c$ is not adjacent to $b$, it cannot reach the filament of $a$ (whose extremities are outside the interval of $b$), which is a contradiction.
\end{proof}
 
We move on to the second inclusion.

\begin{lemma}
$C_a \subseteq$ Interval filaments.
\end{lemma}

\begin{proof}

Consider a graph $G$ in $C_a$. We show how to compute a proper interval filament representation of $G$.
To do so, we first create a set of interval filaments that is not a proper representation of $G$ in general, and then modify it several times, until it becomes a proper intersection representation of $G$. 
The outcome of our construction is illustrated on an example in Figure~\ref{fig:filaments}. 

\begin{figure}[!h]
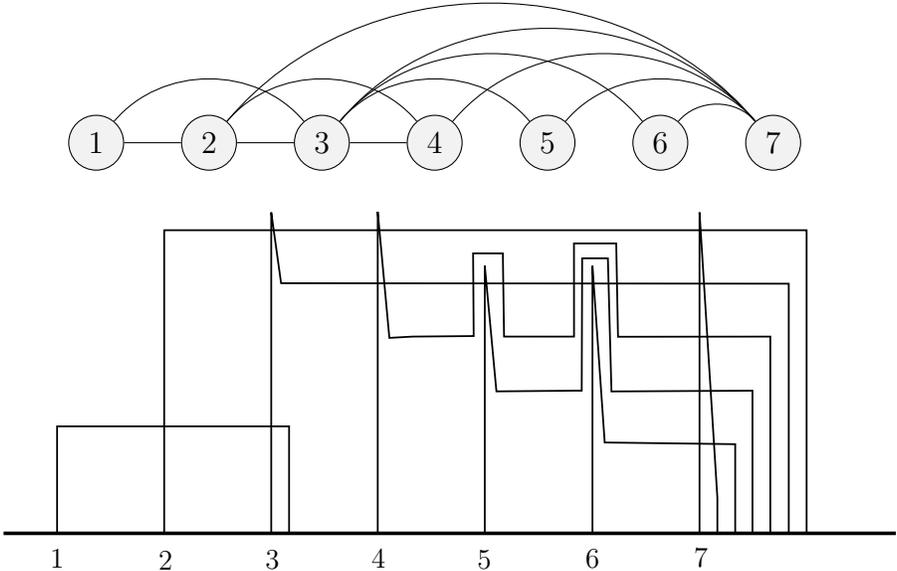

\begin{center}
\input{
			filament-graph.tex}
\vspace{0.5cm}

\scalebox{0.9}{			
\input{fig-filaments}}
\end{center}
\caption{
\label{fig:filaments}
On the top picture, a graph whose vertices are ordered to avoid $P_a$. 
On the bottom picture, the interval filament representation built by our construction. 
(A few modifications of the construction have been used for readability, including using a smaller scale on the vertical axis than on the horizontal axis.)}
\end{figure}

We start with a first representation that we call $R_1$.
Consider an order of the vertices of $G$ avoiding~$P_a$.
We will abuse notation and write $u$ instead of ``the rank of $u$ in the ordering''. 
For every the vertex $u$, we denote its right-most neighbor by $r_u$ (with $r_u=u$ when $u$ has no neighbor on the right).
Let $\varepsilon_u=\frac{n-u+1}{n+1}$. 
We note that for two vertices $u<v$, $0<\varepsilon_v<\varepsilon_u<1$.
In $R_1$, the filament of $u$ starts at position $(u,0)$, goes up to point $(u,r_u-u+\varepsilon_u)$, then goes right to the position $(r_u+\varepsilon_u,r_u-u+\varepsilon_u)$, and then goes down to the grounding line at $(r_u+\varepsilon_u,0)$. Note that at this point each filament forms a square with the grounding line. For further reference, we denote the altitude of horizontal segment as $h_u=r_u-u+\varepsilon_u$. 

\medskip

\begin{claim}
In $R_1$, all the intersections between filaments correspond to edges in the graph, and at any intersection point, there are exactly two filaments intersecting, one that is locally vertical and one that is locally horizontal. 
\end{claim}

\medskip

\begin{claimproof}
Let us start by proving the first part of the claim.
Consider the filaments of two vertices $u<v$. We make a case analysis depending on the placement of $r_u$ and~$r_v$.
\begin{itemize}
\item If $u\leq r_u<v\leq r_v$, then the two filaments do not intersect.
\item If $u<v \leq r_v\leq r_u$, then again the filaments do not intersect. Indeed, the square corresponding to $v$, starts at position $v$ ($>u$) and stops at $r_v+\varepsilon_v$ ($<r_u+\varepsilon_u$), thus is fully included into the square of $u$.
\item If $u<v<r_u<r_v$, then the filaments intersect in $R_1$, but in this case, since by definition $(u,r_u)$ and $(v,r_v)$ are edges in $G$, either $(u,v)$ is an edge in $G$, or we get that $u<v<r_u< r_v$ form a $P_a$ in the ordering, which is a contradiction.
\end{itemize} 

Now, for the rest of the claim, note that any horizontal segment is at an altitude of the form $k+\varepsilon_u$, where $k$ is an integer, and $\varepsilon_u\in (0,1)$ is specific to the node $u$ where the filament starts. Thus, no two horizontal parts of two different filaments can intersect. 
The analogue reasoning holds for the vertical parts. 
The claim follows. 
\end{claimproof}

In general, $R_1$ is not a proper representation of $G$, because there are edges in the graphs that do not correspond to filament intersections.
We now modify this first representation $R_1$ into a second representation $R_2$.

For a vertex $v$, we select (when it exists) the vertex $w$ such that:
\begin{itemize}
\item $w$ is on the left to $v$,
\item $w$ is adjacent to $v$ in $G$,
\item the filaments of $w$ and $v$ do not intersect in $R_1$
\item among the vertices satisfying the three points above, $w$ is the one that has the highest horizontal segment (that is, the largest $h_w$).  
\end{itemize} 

We modify the shape of $v$ with a spike at the beginning: the first vertical line goes up to $h_w$, and then goes down to $h_v$, and continues the shape as before. 
(If we want the filament to not self-intersect, we can go down with some steep enough slope, like in Figure~\ref{fig:filaments}.)

\medskip

\begin{claim}
For every edge of the graph, there is an intersection of filaments in $R_2$. 
Moreover, if the filaments corresponding to some vertices $y<v$ intersect in $R_2$, and $y$ and $v$ are not adjacent in $G$, then there must exist a vertex $x$, with $x<v$, such that $(x,v)\in G$ and $h_y<h_x$. 
\end{claim}

\medskip

\begin{claimproof}
First, note that all filaments intersections from $R_1$ also appear in $R_2$. Also, note that the spikes appear on different zones of $x$-axis, thus they do not interfere.
Now, consider a pair of nodes $v<u$, such that $(v,u)\in G$, but the two corresponding filaments were not intersecting in $R_1$. 
By definition of $r_v$, $u\leq r_v$, thus $v<u <r_v+\varepsilon_v$. 
Therefore, in $R_1$ the horizontal part of the filament of $v$ goes above the left grounding point of $u$.
The vertex $u$ satisfies the three first conditions of the vertex $w$ defined above, and because $h_w$
is chosen to be maximal, we get that the spike of $u$ necessarily touches the filament of $v$.
Thus the intersection between the filaments of $v$ and $u$ appears in $R_2$, which proves the first part of the claim.

Now, if there is an unwanted intersection, that is an intersection in $R_2$ that does not appear in $G$, it must be because of a spike of some vertex $v$, that intersects the horizontal part of the filament of a vertex $y$ with $(v,y)\notin E$, while going up to touch~$u$. 
This can only happen if the conditions of the claim are satisfied. 
\end{claimproof}

We now define $R_3$, which will be our final filament representation.
Let us first define an $\varepsilon$-\emph{chimney}. 
Consider a horizontal segment $[(x_1,y),(x_2,y)]$ and a vertical segment $[(x,y_1),(x,y_2)]$ that intersect at position $(x,y)$. 
Then an $\varepsilon$-chimney is a modification of the horizontal segment, to avoid the vertical segment. More precisely, it is the piece-wise segment going through the following points: $(x_1,y), (x-\varepsilon,y), (x-\varepsilon, y_2+\varepsilon), (x+\varepsilon, y_2+\varepsilon), (x+\varepsilon, y), (x_2,y)$.

Let us call \emph{a conflict} an intersection between two filaments that does not correspond to an edge in $G$.
As mentioned earlier, a conflict is necessarily between the spike of a filament and the horizontal plateau of another filament. 
For any given filament $u$, we show how to solve all the conflicts that involve the spike of $u$, without loosing the intersection that are needed in $G$. 
We do so by only modifying the filaments in a small vertical column around the abscissa of $u$. 
We will maintain the invariant that when we consider a filament $v$, every filament~$w$ above it has been treated: either $(u,w)\in E$ and the intersection has been preserved, or $(u,w)\notin E$ and the intersection has been removed.  
We consider the filaments that intersect the spike of $u$ ordered by the altitude of their horizontal plateaus, from the highest to the lowest.
For each such filament $v$, we do the following:
\begin{itemize}
\item If $(u,v)\in E$, do nothing. 
\item If $(u,v)\notin E$:
\begin{itemize} 
\item If there is a filament $w$ above $v$, such that $(v,w)\notin E$, and $w$ intersects the spike of $u$, then we claim that there is a contradiction. 
Indeed, if $(v,w)$ is not in~$G$, and  $h_v<h_w$, necessarily, $w<v<u<r_v<r_w$. Then $w<v<u<r_v$ is an occurrence of $P_a$, because $(w,v)\notin E$ (by hypothesis), $(w,u)\in E$ (by the invariant) and $(v,r_v)\in E$ (by definition).   
\item If there is no filament $w$ above $v$, such that $(v,w)\notin E$, and $w$ intersects the spike of $u$, then we transform the filament of $v$, by building a chimney around the spike of $u$. 
The $\varepsilon$ of this chimney is taken to be smaller than all the $\varepsilon$ used for chimneys of filaments modified so far around the spike of $u$. 
If this is the first chimney, then we take $\varepsilon$ to be the minimum between 1/4, and the space between the top of the spike of $u$ and the next plateau above (if it exists).  
\end{itemize}
\end{itemize}

We claim that this construction gives a proper interval filament representation of~$G$. 
First note that, the conflict resolution is independent at all vertices, because the chimney we build are taken with a sufficiently small $\varepsilon$ around the spikes. 
Hence, it is sufficient to show for each spike that: we have not lost correct intersections, we have removed all the incorrect intersections, and have not created new ones.

The correct intersections that were present in $R_2$ are of two types: the ones coming from~$R_1$ and the ones between a spike and horizontal plateaus. 
We have only modified the graphs around the spikes, thus the first type is safe. 
Moreover, when we had an intersection corresponding to an edge $(u,v)\in E$, we have not modified the filament, hence the second type is also safe.

All the incorrect intersections at some given spike in $R_2$ have been removed thanks to our chimney constructions.

Finally, we have not created any new intersection. Indeed, while building a chimney, the only new incorrect intersection that we might have created is with another chimney, but this cannot happen as we have taken care of reducing the $\varepsilon$ parameter at every chimney.
\end{proof}

\section{Grounded stairs graphs between $C_{ab}$ and $C_{abc}$}
\label{sec:grounded-stairs}

\begin{Theo}\label{thm:grounded-stairs} The following holds:
 $C_{ab} \subseteq$ grounded stairs $\subseteq C_{abc}$. 
\end{Theo} 

The main part of this theorem is the first inclusion. 
The second inclusion can actually be derived from a result in \cite{PachT20}. More precisely, in \cite{PachT20} it is proved that the disjointness graph of $x$-monotone curves are semi-comparability graphs. This is related to our result because: the disjointness graph is the complement of the intersection graph, $x$-monotone graphs are a generalization of grounded stairs graphs, and semi-comparability are the graphs that avoid the complement of $C_{abc}$.

\begin{proof}
Let us start with the first inclusion, and consider a graph $G=(V,E)$ in $C_{ab}$. 
We build a grounded stairs representation incrementally, by adding the stairs one by one, from left to right, in the ordering of $G$ avoiding the pattern~$P_{a,b}$. 
Actually, for every new vertex $i$, allow ourselves to modify the stairs 1 to $i-1$, such that at the end of the process, the adjacency of the vertices 1 to $i$ in the graph is the same as the one in the grounded stairs configuration. 

For any new vertex $i$, we create an infinite vertical line, grounded at position $i$. We will cut the line into a finite segment at the end of this step of the construction, and maintain the invariant that all the stairs are finite.
We will maintain the invariant that the representation of vertices 1 to $k$ only uses abscissas in $[1,k]$, thus, at that point, the line has no intersection. 
Consider a stair $j$, with $j<i$ in the ordering, and such that $(i,j)\in E$. 
We claim that it is possible to modify the stairs of $j$, so that they intersect the segment of $i$, without creating an unwanted intersection with the stairs of any vertex $k$, with $(j,k)\notin E$.
Note that proving that we can make the stairs of $j$ and $i$ intersect simply means that we can extend the stairs of $j$ beyond the abscissa $i$.
 
For the sake of contradiction, suppose there are some stairs $j$, such that the extension to $i$ is no possible.
It means that there is a position $x$ on the horizontal axis, such that the stair $j$ cannot be extended beyond abscissa $x$. 
We extend the stairs $j$ to a position $p=(x-\varepsilon,t)$ for a small $\varepsilon$. 
Now, as by invariant, all the stairs of index smaller than $i$ are finite, there must be an ordinate $y$, such that if we make stairs $j$ go up from $p$, it cannot go further than $y$. (If such an ordinate $y$ would not exist, then we could just go high enough to be above all the stairs except $i$, and then go straight to $i$, which would contradict that fact that we cannot go beyond abscissa $x$.)

In this situation, the position $(x,y)$ must be the intersection of two stairs $\ell$ and $r$, $\ell$ having an horizontal segment at height $y$ , and $r$ having a vertical segment at abscissa $x$. And necessarily $\ell$ and $r$ are not adjacent to $j$ in $G$.
We claim that in the ordering, $\ell < j < r <i $. 
Indeed, $i$ has the right-most stairs by definition, and given the configuration described above, any other ordering would imply an intersection between the stairs $\ell$ or $r$, and the stairs of $j$. 
Then, the vertices $\ell < j < r <i $ correspond the pattern $P_{a,b}$ because we have edges $(\ell,r)$ and $(j,i)$, and non-edges $(\ell,j)$ and $(j,r)$. This is a contradiction. Therefore, we can always extend the stairs of $j$ to the stairs of $i$, and get the same adjacency in the representation and in the graph for vertices from 1 to $i$. Finally, we cut the infinite vertical line of $i$ just above its last intersection to make it a finite segment. This finishes the proof of the first inclusion.

Let us now prove that grounded stairs graphs are included in $C_{abc}$. 
Consider a grounded stairs representation. 
We claim that the grounding ordering of the associated graph avoids the pattern~$P_{abc}$. 
Suppose that it is not the case, and that the pattern appears on some vertices $a<b<c<d$. 
Consider the zone of the plane between the grounding line, the stairs of $a$, and the stairs of $c$, below their (first) intersection, whose coordinates will be denoted $(x,y)$.
This zone contains the grounding point of $b$. 
Because intersections between the stairs of $b$ and the stairs of $a$ and $c$ are forbidden, the stairs of $b$ are fully contained in this zone.
Now, the stairs of $d$ must appear strictly on the right of the stairs of $c$ until ordinate $y$, because these two do not intersect either. 
Finally, the rest of the stairs of $d$ is in the quarter of the plane that is on the top right of $(x,y)$. For both parts of the stairs of $d$, the intersection with the stairs of $b$ is impossible. This is a contradiction, and proves the inclusion.
\end{proof}

\section{Inclusions of grounded convex and string}
\label{sec:convex-string}

In this section, we show two easy inclusions of geometric classes into pattern classes. 

\begin{Theo}
\label{thm:convex-Cabc}
Grounded convex graphs are included in $C_{abc}$.
\end{Theo}

\begin{proof}
Consider the grounded convex representation, and suppose that the pattern of $C_{abc}$ is present with vertices $a<b<c<d$, when considering  the graph in the grounding ordering. 
Note first that the intersection of the convex shapes $a$ and $c$ is itself a convex shape. 
Let $P$ be a lowest point of this convex intersection.
By similar arguments as in the previous section, the ordinate of $P$ is the highest ordinate that the shape of $b$ can reach.
Now, as the shape $d$ must avoid the shape $c$, it must in particular pass ``above'' the segment between $P$ and the grounding point of $c$. 
Therefore $b$ and $d$ can only touch at $P$, and this is impossible as we required that there cannot be a contact of more than two shapes on each point (except for interval graphs). 
This is contradiction, it proves the theorem. 
\end{proof}

\begin{Theo}
\label{thm:grounded-strings}
Grounded string graphs are included in $C_{abcd}$.
\end{Theo}

This statement was already proved in \cite{PachT19}, and it corresponds to the proof sketch we gave in Section~\ref{sec:overview}. We reprove it here formally for completeness.

\begin{proof}
We again reason by contradiction, and consider the grounding ordering in a representation by grounded strings. Suppose we have four vertices $a<b<c<d$, such that the only edges between these vertices are $(a,c)$ and $(b,d)$. 
Then, consider the closed region of the plane that is inside the Jordan curve defined by the string of $a$, the string of  $c$, and the grounding line. 
The shape $b$ must be inside this region because its root is, and it cannot intersect any of the boundaries, by Jordan curve theorem. 
On the other end $d$ must be outside this region, because likewise its root is outside and it cannot cross the curve. Then $b$ and $d$ cannot intersect, which is a contradiction. 
\end{proof}

Note that another way to see Theorem~\ref{thm:grounded-strings} is that the complement of grounded string graphs avoids the complement of $P_{abcd}$, which is an induced cycle in the natural order. 
This is actually not very surprising, as it is known that the class of complement of triangle-free intersection graphs of grounded curves is exactly the class of Hasse diagrams~\cite{MiddendorfP93}, 
and that even without the triangle-free assumption, these complement graphs avoid an infinite family of ordered cycles~\cite{Sinden66}. 

\section{Strict inclusions and open problems}
\label{sec:imcomparable}

In this section, we focus on inequalities between included classes, and open problems.
In particular, we finish giving the results that are used to establish the diagram of Figure~\ref{fig:diagram}, given in the overview section of the paper.

We start with new inequalities between pattern classes.
Remember that inclusion of patterns imply inclusion of classes (Observation~\ref{obs:inclusion-patterns}), thus we have the following inclusions: $C_{\emptyset} \subseteq C_a  \subseteq C_{ab} \subseteq C_{abc} \subseteq C_{abcd}\subseteq ALL$, where $ALL$ stands for the set of all graphs. Also $C_{\emptyset} \subseteq C_b  \subseteq C_{ab}$.

Two inclusions are known to be strict:

\begin{Obs}
\label{obs:Cempty-Ca-Cb}
The class $C_{\emptyset}$, is strictly included into $C_a$ and $C_b$. 
For example $K_4$ has no ordering avoiding $P_{\emptyset}$, but all orderings avoid $P_a$ and $P_b$.
\end{Obs}

We prove three other inequalities between included classes.

\begin{Theo}
\label{thm:pattern-strict}
$C_a \subsetneq C_{ab}$, $C_b \subsetneq C_{ab}$, and $C_{abcd} \subsetneq ALL$
\end{Theo}

\begin{proof}
The proof of this theorem uses a computer program for enumerating vertex orderings. The source code  is available upon request.

For the first two inequalities, $C_a \subsetneq C_{ab}$ and $C_b \subsetneq C_{ab}$, we consider the graph of Figure~\ref{Cab-Ca}. 
This graph is in $C_{ab}$, as the second pictures shows, and a computer enumeration shows that it is not in $C_{a}$, nor in $C_{b}$. 

\begin{figure}[!h]
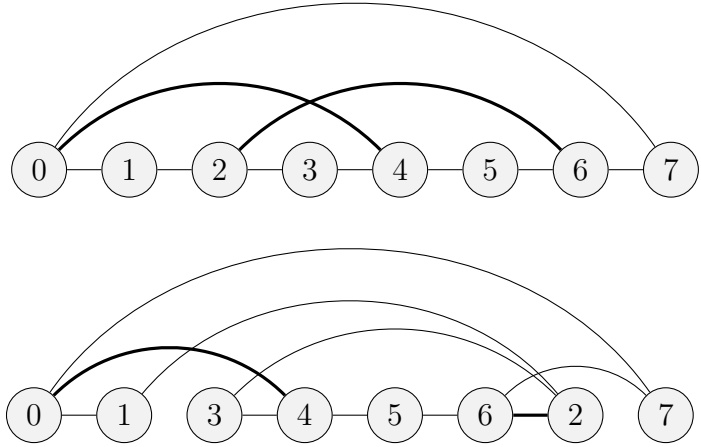

\begin{center}
\input{strict-a-ab.tex}
\input{strict-a-ab-2.tex}
\caption{\label{Cab-Ca}
Two vertex ordering of a graph that is in $C_{ab}$, but not in $C_a$, nor in $C_b$. The first pictures illustrates how we discovered this graph: taking $P_{\emptyset}$ and then forcing some order conditions with intermediate nodes and a cycle. The second picture is an ordering avoiding $P_{ab}$. (We made some edges thicker to increase readability.)}
\end{center}
\end{figure}

For the third inequality, $C_{abcd} \subsetneq ALL$, we consider the graph of Figure~\ref{Cabcd}. 
Again a computer enumeration shows that no vertex ordering of this graph avoids the pattern $P_{abcd}$, thus the graph is not in $C_{abcd}$.

\begin{figure}[!h]
\begin{center}
\input{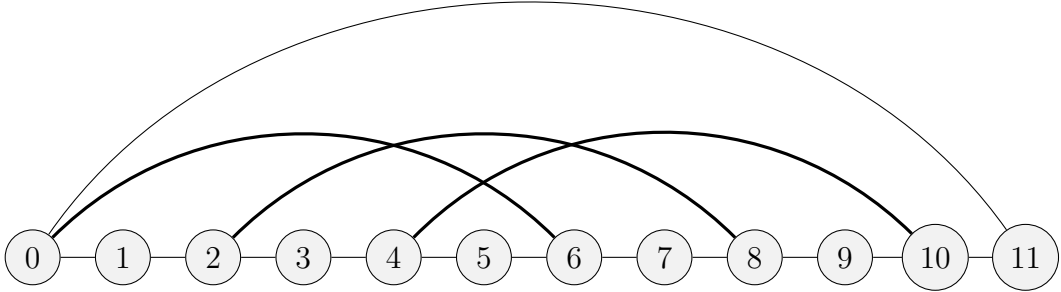}
\caption{\label{Cabcd}A graph that is not in $C_{abcd}$, by computer enumeration. (We made some edges thicker to increase readability.)}
\end{center}
\end{figure}
\end{proof}

A natural open problem here is whether all the inclusions are strict, and even more generally, the following question.

\begin{open}
Let $P$ be a pattern with a non-decided edge $(x, y)$. Let $P'$ be the pattern obtained from $P$ by adding the edge $(x,y)$. Under which conditions  does $C_P \subsetneq C_{P'}$ hold? What about adding a non-edge $(x,y)$?
\end{open}

We now give a list of inequalities between classes, that are either known, or easy to derive from known results.

\begin{Obs}
\label{obs:inequalities}
\begin{enumerate}
\item
\label{item:L-seg}
Grounded L-shape graphs are strictly included in grounded segment graphs~\cite{MiddendorfP92}.
\item 
\label{item:GR-GL-Cab}
Grounded L-shape graphs and grounded rectangle graphs are incomparable~\cite{JelinekT19}. 
This result has several consequences. 
First the inclusions of grounded rectangle graphs grounded L-shape graphs into $C_{ab}$ (which follow from the pattern inclusions) are strict. 
Second, the inclusion of grounded rectangle graphs into grounded convex graphs (which follows from the comparison of the shapes) is also strict.
\item 
\label{item:rectangle-segments}
Grounded rectangle graphs and grounded segment graphs are incomparable~\cite{JelinekT19}. As a consequence, the inclusion of grounded segment graphs into grounded convex graphs (which follows from the comparison of the shapes) is strict.
\item 
\label{item:forest-outerplanar}
We know from the literature (or alternatively from Theorem~\ref{thm:touching-L-shapes} and~\ref{thm:touching-strings}) that forest are included in outerplanar graphs. This inclusion is strict as witnessed by the cycle on four vertices, $C_4$, that is outerplanar but not acyclic.
\item 
\label{item:forest-GL}
By Theorem~\ref{thm:touching-L-shapes}, we get that forests are included in grounded L-shape graphs. 
This inclusion is strict, as witnessed by $C_4$ again (Figure~\ref{fig:L-shapes-cycle} shows that $C_4$ is a grounded L-shape graph).
\item 
\label{item:interval-inclusions}
Interval graphs are included in $C_a$, cocomparability graphs, grounded rectangle graphs, and grounded L-shape, as can be seen from the pattern inclusions (for the first three class) and from the generality of the shapes and grounding (for the fourth).
These inclusions are strict as witnessed by $C_4$ again.
\item 
\label{item:permutation-GL-GS}
The class of 2-grounded segments graphs (permutation graphs) is strictly included in the classes of grounded L-shapes and grounded segments~\cite{JelinekT19}.
\item 
\label{item:cocomp-filament}
Cocomparability graphs are included in interval filament graphs, by Observation~\ref{obs:two-lines-one-line}, and the inclusion is strict, as witnessed by $C_6$ which is not a cocomparability graph~\cite{Gallai67} and can be easily represented by interval filaments.
\end{enumerate}
\end{Obs}

Now that we have all the results that are used to build the Diagram of Figure~\ref{fig:diagram}, we can list some open problems. 
Of course settling the status of each inclusion (either equality or strict inclusion) is a natural direction. For example, the lowest such edge in the diagram is for the inclusion of $C_a$ in interval filament graphs. 
This question would be solved if could answer positively the following, innocent-looking problem.

\begin{open}
Are $C_a$ and cocomparability graphs incomparable? (Or are cocomparability graphs included in $C_a$, the reverse being impossible, as witnessed by $C_6$?)
\end{open}

We know that $C_{\emptyset}$ is the class of outerplanar graphs, and that $C_b$ is the class of grounded rectangle graphs, which gives us several points of view on these two classes defined by forbidden patterns. We lack such diversity for other the other pattern classes considered in the paper.

\begin{open}
Can we characterize $C_{a, b}$, $C_{a, b, c}$ and $C_{a, b, c , d}$, without patterns? Do these classes have geometric characterizations?
\end{open}

\section{Complexity of recognition}
\label{sec:recognition}

As said in the introduction, one of the motivations to study the relation between patterns and grounded intersection graphs is that it is a way to better understand patterns on four vertices. 
In this section, we survey the known results and list open problems about the complexity of the recognition of the classes that are defined by patterns on four vertices. 

Let us start by classes that are known to be recognizable in polynomial time. Remember that all the classes characterized by patterns on three vertices can be recognized in polynomial time \cite{HellMR14, FeuilloleyH21}. When we go to four vertices, the landscape is more complex, but there are still ``good cases''.

First, \emph{outerplanar graphs}, that have been mentioned a lot already, as they correspond to~$C_{\emptyset}$, can be recognized in linear time \cite{Mitchell79}.
Second, let us consider the class of \emph{cographs}, which is the smallest class of graphs that contains the unique vertex graph and is closed under series and parallel compositions. 
The cographs are also the $P_4$-free graphs \cite{Seinsche74}, hence they can be characterized by the existence of a vertex ordering avoiding the twelve orderings of $P_4$. 
But there is a more compact set of forbidden patterns~\cite{Damaschke90}, see Fig.~\ref{fig:cographs}.
The cographs can also be recognized in linear time \cite{CorneilPS85}.

\begin{figure}[!h]
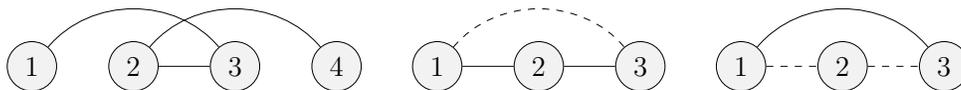

\centering
	\begin{tabular}{ccc}
	\scalebox{0.9}{
	\input{FoldedP4}}					&
	\scalebox{0.9}{
	\input{comparability}}&
	\scalebox{0.9}{
	\input{cocomparability}}
	\end{tabular}
	\caption{
	\label{fig:cographs}
	Forbidden patterns of cographs.}
\end{figure}

A less-known class is the \emph{strongly chordal graphs} in which every even cycle of length greater or equal to 6 admits an odd chord.  
These are characterized by the existence of a strong perfect  elimination ordering of the vertices, which corresponds to the patterns of Figure~\ref{fig:strongly-chordal}~\cite{Farber83}. 
They can be recognized in time $O(\min(n^2, (n + m) \log n))$~\cite{S93}.

\begin{figure}[!h]
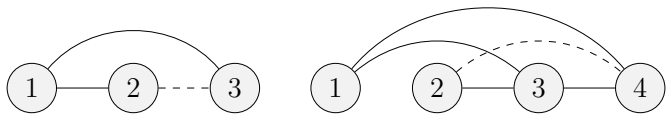

	\centering
	\begin{tabular}{cc}
	\scalebox{0.9}{
	\input{cott-1}}
	&
	\scalebox{0.9}{
	\input{strongly_chordal}}
	\end{tabular}
	\caption{\label{fig:strongly-chordal}Forbidden patterns of strongly chordal graphs.}
\end{figure}

The \emph{complement of threshold tolerance graphs (coTT for short)} is surclass of the strongly chordal graphs that have been characterized with the patterns of Figure~\ref{fig:coTT}~\cite{MonmaRT88}, and whose recognition is also polynomial \cite{GolovachHLMSSS17}.

\begin{figure}[!h]
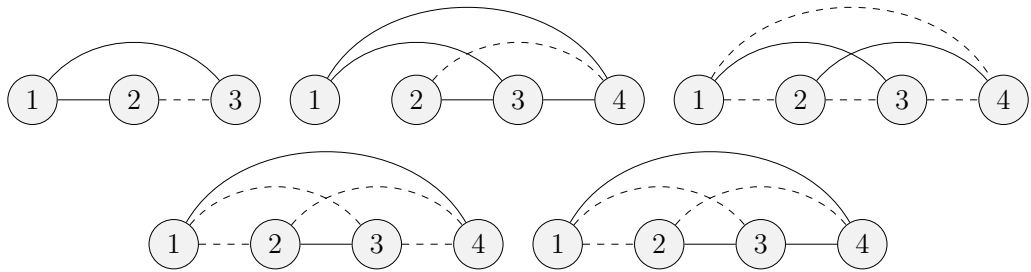

	\centering
	\scalebox{0.9}{
	\input{cott-1}}
	\scalebox{0.9}{
	\input{cott-2}}
	\scalebox{0.9}{	
	\input{cott-3}}
	\scalebox{0.9}{
	\input{cott-4}}
	\scalebox{0.9}{	
	\input{cott-5}}		
	\caption{\label{fig:coTT}
	Forbidden patterns of coTT graphs.}
\end{figure}

Finally, note that \emph{$K_4$-free graphs} (resp. their complement) are characterized by the pattern that is an ordered version of $K_4$ (resp. of $\bar{K_4}$).  
These can be recognized in time $O(m^{\frac{\alpha +1}{2}})$, where $\alpha$ is the best exponent for boolean matrix multiplication~\cite{KKM95}. 

\medskip

As hinted before, there are also classes that are characterized by patterns on four vertices and that are hard to recognize. The simplest example of this phenomenon is \emph{3-colorable graphs}, whose recognition is a classic NP-hard problem. 
These are characterized by the forbidden pattern that is simply a path on four vertices (see \emph{e.g.} \cite{FeuilloleyH21}, but this is also a special case of the Gallai-Hasse-Roy-Vitaver theorem). 

Another example is the class of \emph{graphs with queue number 1}, that can be characterized with the pattern of Figure~\ref{fig:queue-number-one}~\cite{HeathLR92}, and whose recognition is NP-hard~\cite{HeathR92}.

\begin{figure}[!h]
\centering
	\scalebox{0.9}{
	\input{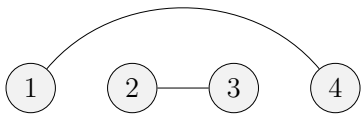}}
\caption{\label{fig:queue-number-one} Forbidden pattern of the graphs of queue number 1.}
\end{figure}

One more example is \emph{perfectly orderable graphs}, that are the graphs that admit a perfect ordering of the vertices, i.e. an ordering where there is no path on four vertices $(a-b-c-d)$, such that $a<b$ and $d< c$. 
These are characterized by the forbidden patterns of Figure~\ref{fig:perfectly-orederable}, and their recognition was proven to be NP-hard in \cite{MiddendorfP90}.

\begin{figure}[!h]
	\centering
	\scalebox{0.9}{
	\input{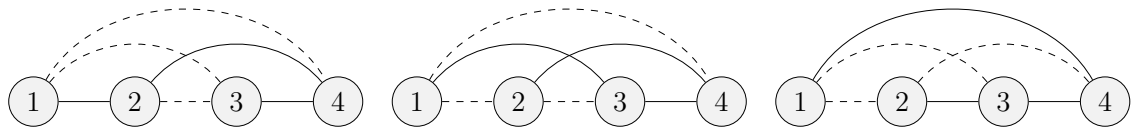}}
	\caption{\label{fig:perfectly-orederable}Forbidden patterns of perfectly orderable graphs.}
\end{figure}

Let us now move on to the interesting classes for which we do not know the complexity of the recognition. First, the recognition complexity of almost all the pattern classes on Figure~\ref{fig:diagram} is unknown. In particular, the question for grounded rectangle graphs, that is $C_b$, is still open, although it has been raised several time in the literature. 
A related class of interest is the one where the pattern has edges $(1,3)$, $(2,3)$ and $(2,4)$. This pattern correspond to $P_b$ but with an edge instead of a non-edge, and it is also a permutation of the pattern for 3-colorable graph. 
To our knowledge this class has not received attention yet. 

We will finish this section with graph classes related to two parameters: the LMIM-width and the MIM-width. 
These are interesting parameters based on the size of some maximum matchings, that are related to the other classic parameters such as treewidth, branchwidth, rankwidth and cutwidth~\cite{V12, KKST17}. 
The graph of LMIM-width 1 are known to be characterized by the two patterns of Figure~\ref{fig:LMIM}, including $P_{abcd}$. 

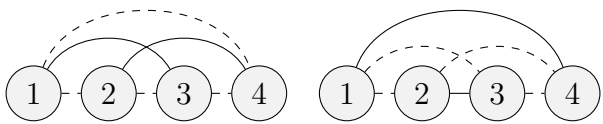
\begin{figure}[h!]
\begin{center}
\begin{tabular}{cc}
\begin{tikzpicture} 
[scale=1.0,auto=left]   
\node[circle,draw,fill=black!5] (n1) at (0,0) {1};   
\node[circle,draw,fill=black!5] (n2) at (1,0) {2};   
\node[circle,draw,fill=black!5] (n3) at (2,0)  {3};   
\node[circle,draw,fill=black!5] (n4) at (3,0)  {4};

\draw (n1) to[bend left=60] (n3); 
\draw (n2) to[bend left=60] (n4);
\draw[dashed] (n1) to (n2);
\draw[dashed] (n2) to (n3);
\draw[dashed] (n3) to (n4);
\draw[dashed] (n1) to[bend left=70] (n4);
\end{tikzpicture}
&
\begin{tikzpicture}   
[scale=1.0,auto=left,every node/.style={circle,draw,fill=black!5}]   
\node (n1) at (0,0) {1};   
\node (n2) at (1,0) {2};   
\node (n3) at (2,0)  {3};   
\node (n4) at (3,0)  {4};

\draw (n1) to[bend left=70] (n4); 
\draw (n2) to (n3);
\draw[dashed] (n1) to (n2);
\draw[dashed] (n3) to (n4);
 \draw[dashed] (n1) to[bend left=50] (n3);
\draw[dashed] (n2) to[bend left=50] (n4);
\end{tikzpicture}
\end{tabular}
\end{center}
\caption{\label{fig:LMIM}The patterns characterizing the graphs of LMIM-width  1. The first pattern is~$P_{abcd}$.}
\end{figure}

The graphs that have MIM-width 1 correspond to linear Hsu-decomposable graphs, and they can be characterized by the two forbidden patterns of Figure~\ref{fig:MIM}~\cite{V12, KKST17}. 

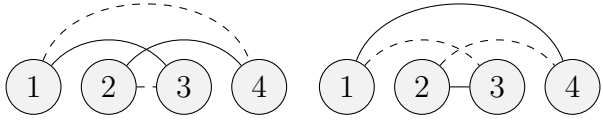
\begin{figure}[h!]
\begin{center}
\begin{tabular}{cc}
\begin{tikzpicture}   
[scale=1.0,auto=left]   
\node [circle,draw,fill=black!5](n1) at (0,0) {1};   
\node[circle,draw,fill=black!5] (n2) at (1,0) {2};   
\node[circle,draw,fill=black!5] (n3) at (2,0)  {3};   
\node[circle,draw,fill=black!5] (n4) at (3,0)  {4};       

\draw (n1) to[bend left=50] (n3); 
\draw (n2) to[bend left=50] (n4);
\draw[dashed] (n2) to (n3);
 \draw[dashed] (n1) to[bend left=70] (n4);
\end{tikzpicture}
&
\begin{tikzpicture}  
[scale=1.0,auto=left,every node/.style={circle,draw,fill=black!5}]   
\node (n1) at (0,0) {1};   
\node (n2) at (1,0) {2};   
\node (n3) at (2,0)  {3};   
\node (n4) at (3,0)  {4};

\draw (n1) to[bend left=70] (n4); 
\draw (n2) to (n3);
\draw[dashed] (n2) to[bend left=50] (n4);
\draw[dashed] (n1) to[bend left=50] (n3);
\end{tikzpicture}
\end{tabular}
\end{center}
\caption{\label{fig:MIM} The patterns characterizing the graphs of MIM-width 1. The first pattern is~$P_{bd}$.}
\end{figure}

As far as we know the complexity of the recognition of LMIM-width  1 and MIM-width 1 are still widely open. 


\vspace{0.5cm}
\noindent
\textbf{Conclusion} 

In this paper we prove that forbidden patterns can help to better understand the relationships between graph classes, but it is no clear how they could help for the classification of the complexity of the recognition of the corresponding graph classes.
Since very similar patterns may lead to either linear (such as $C_{\emptyset}$) or  perhaps NP-hard recognition algorithms (such as $C_b$).



\paragraph{Acknowledgments.}
The authors thank Istvan Tomon for pointing out to references~\cite{PachT19}, \cite{PachT20}, \cite{MiddendorfP93} and \cite{Sinden66}.

\newpage{}
\bibliography{biblio.bib}{}

\begin{thebibliography}{10}

\bibitem{CabelloJ17}
Sergio Cabello and Miha Jejcic.
\newblock Refining the hierarchies of classes of geometric intersection graphs.
\newblock {\em Electron. J. Comb.}, 24(1):P1.33, 2017.

\bibitem{CardinalFMTV18}
Jean Cardinal, Stefan Felsner, Tillmann Miltzow, Casey Tompkins, and Birgit
  Vogtenhuber.
\newblock Intersection graphs of rays and grounded segments.
\newblock {\em J. Graph Algorithms Appl.}, 22(2):273--295, 2018.

\bibitem{CatanzaroCFHHHS15}
Daniele Catanzaro, Steven Chaplick, Stefan Felsner, Bjarni~V Halld{\'o}rsson,
  Magn{\'u}s~M Halld{\'o}rsson, Thomas Hixon, and Juraj Stacho.
\newblock Max point-tolerance graphs.
\newblock {\em Discrete Applied Mathematics}, 2015.

\bibitem{Corneil97}
Derek~G Corneil.
\newblock Extensions of permutation and interval graphs.
\newblock In {\em Proc. 18th Southeastern Conference on Combinatorics, Graph
  Theory and Computing}, pages 267--276, 1987.

\bibitem{CorneilOS09}
Derek~G. Corneil, Stephan Olariu, and Lorna Stewart.
\newblock The {LBFS} structure and recognition of interval graphs.
\newblock {\em {SIAM} J. Discrete Math.}, 23(4):1905--1953, 2009.

\bibitem{CorneilPS85}
Derek~G. Corneil, Yehoshua Perl, and Lorna~K. Stewart.
\newblock A linear recognition algorithm for cographs.
\newblock {\em {SIAM} J. Comput.}, 14(4):926--934, 1985.

\bibitem{CorreaFPS15}
Jos{\'{e}}~R. Correa, Laurent Feuilloley, Pablo P{\'{e}}rez{-}Lantero, and
  Jos{\'{e}}~A. Soto.
\newblock Independent and hitting sets of rectangles intersecting a diagonal
  line: Algorithms and complexity.
\newblock {\em Discrete {\&} Computational Geometry}, 53(2):344--365, 2015.

\bibitem{DaganGP88}
Ido Dagan, Martin~Charles Golumbic, and Ron~Y. Pinter.
\newblock Trapezoid graphs and their coloring.
\newblock {\em Discret. Appl. Math.}, 21(1):35--46, 1988.

\bibitem{Damaschke90}
P.~Damaschke.
\newblock Forbidden ordered subgraphs.
\newblock {\em Topics in Combinatorics and Graph Theory, R. Bodendiek and R.
  Hennm Eds}, pages 219--229, 1990.

\bibitem{DaviesKMW21}
James Davies, Tomasz Krawczyk, Rose McCarty, and Bartosz Walczak.
\newblock Grounded l-graphs are polynomially {\(\chi\)}-bounded.
\newblock {\em CoRR}, abs/2108.05611, 2021.

\bibitem{KKST17}
D.Y.Kang, O.Kwon, T.Strømme, and J.A.Telle.
\newblock A width parameter useful for chordal and co-comparability graphs.
\newblock {\em {Theoretical Computer Science}}, 704:1--17, 2017.

\bibitem{EvenPL72}
Shimon Even, Amir Pnueli, and Abraham Lempel.
\newblock Permutation graphs and transitive graphs.
\newblock {\em J. {ACM}}, 19(3):400--410, 1972.

\bibitem{Gavril00}
Gavril Fanica.
\newblock Maximum weight independent sets and cliques in intersection graphs of
  filaments.
\newblock {\em Inf. Process. Lett.}, 73(5-6):181--188, 2000.

\bibitem{Farber83}
Martin Farber.
\newblock Characterizations of strongly chordal graphs.
\newblock {\em Discrete Mathematics}, 43(2-3):173--189, 1983.

\bibitem{FelsnerMW97}
Stefan Felsner, Rudolf M{\"{u}}ller, and Lorenz Wernisch.
\newblock Trapezoid graphs and generalizations, geometry and algorithms.
\newblock {\em Discret. Appl. Math.}, 74(1):13--32, 1997.

\bibitem{FeuilloleyH21}
Laurent Feuilloley and Michel Habib.
\newblock Graph classes and forbidden patterns on three vertices.
\newblock {\em SIAM Journal on Discrete Mathematics (SIDMA)}, 35(1):55--90,
  2021.

\bibitem{Gallai67}
Tibor Gallai.
\newblock Transitiv orientierbare graphen.
\newblock {\em Acta Mathematica Hungarica}, 18(1-2):25--66, 1967.

\bibitem{GolovachHLMSSS17}
Petr~A. Golovach, Pinar Heggernes, Nathan Lindzey, Ross~M. McConnell,
  Vin{\'{\i}}cius~Fernandes dos Santos, Jeremy~P. Spinrad, and Jayme~Luiz
  Szwarcfiter.
\newblock On recognition of threshold tolerance graphs and their complements.
\newblock {\em Discret. Appl. Math.}, 216:171--180, 2017.

\bibitem{GolumbicRU83}
Martin~Charles Golumbic, Doron Rotem, and Jorge Urrutia.
\newblock Comparability graphs and intersection graphs.
\newblock {\em Discrete Mathematics}, 43(1):37--46, 1983.

\bibitem{HeathLR92}
Lenwood~S. Heath, Frank~Thomson Leighton, and Arnold~L. Rosenberg.
\newblock Comparing queues and stacks as mechanisms for laying out graphs.
\newblock {\em {SIAM} J. Discrete Math.}, 5(3):398--412, 1992.

\bibitem{HeathR92}
Lenwood~S. Heath and Arnold~L. Rosenberg.
\newblock Laying out graphs using queues.
\newblock {\em {SIAM} J. Comput.}, 21(5):927--958, 1992.

\bibitem{HellMR14}
Pavol Hell, Bojan Mohar, and Arash Rafiey.
\newblock Ordering without forbidden patterns.
\newblock In {\em Algorithms - {ESA} 2014 - 22th Annual European Symposium,
  Wroclaw, Poland, September 8-10, 2014. Proceedings}, pages 554--565, 2014.

\bibitem{Hixon13}
Thomas~Stuart Hixon.
\newblock Hook graphs and more : Some contributions to geometric graph theory.
\newblock Master's thesis, Technische Universitat Berlin, 2013.

\bibitem{JelinekT19}
V{\'{\i}}t Jel{\'{\i}}nek and Martin T{\"{o}}pfer.
\newblock On grounded {L}-graphs and their relatives.
\newblock {\em Electron. J. Comb.}, 26(3):P3.17, 2019.

\bibitem{KKM95}
T.~Kloks, D.~Kratsch, and H.~M\"uller.
\newblock Finding and counting small induced subgraphs efficiently.
\newblock In {\em Lecture Notes in Comp. Sci. 1017}, 1995.

\bibitem{KostochkaK97}
Alexandr~V. Kostochka and Jan Kratochv{\'{\i}}l.
\newblock Covering and coloring polygon-circle graphs.
\newblock {\em Discret. Math.}, 163(1-3):299--305, 1997.

\bibitem{LekkeikerkerB62}
C~Lekkeikerker and J~Boland.
\newblock Representation of a finite graph by a set of intervals on the real
  line.
\newblock {\em Fundamenta Mathematicae}, 51(1):45--64, 1962.

\bibitem{McGuinness96}
Sean McGuinness.
\newblock On bounding the chromatic number of l-graphs.
\newblock {\em Discret. Math.}, 154(1-3):179--187, 1996.

\bibitem{Mertzios12}
George~B. Mertzios.
\newblock The recognition of triangle graphs.
\newblock {\em Theor. Comput. Sci.}, 438:34--47, 2012.

\bibitem{Mertzios15}
George~B. Mertzios.
\newblock The recognition of simple-triangle graphs and of linear-interval
  orders is polynomial.
\newblock {\em {SIAM} J. Discret. Math.}, 29(3):1150--1185, 2015.

\bibitem{MiddendorfP90}
Matthias Middendorf and Frank Pfeiffer.
\newblock On the complexity of recognizing perfectly orderable graphs.
\newblock {\em Discrete Mathematics}, 80(3):327--333, 1990.

\bibitem{MiddendorfP92}
Matthias Middendorf and Frank Pfeiffer.
\newblock The max clique problem in classes of string-graphs.
\newblock {\em Discret. Math.}, 108(1-3):365--372, 1992.

\bibitem{MiddendorfP93}
Matthias Middendorf and Frank Pfeiffer.
\newblock Weakly transitive orientations, hasse diagrams and string graphs.
\newblock {\em Discret. Math.}, 111(1-3):393--400, 1993.

\bibitem{Mitchell79}
Sandra~L. Mitchell.
\newblock Linear algorithms to recognize outerplanar and maximal outerplanar
  graphs.
\newblock {\em Inf. Process. Lett.}, 9(5):229--232, 1979.

\bibitem{MonmaRT88}
Clyde~L. Monma, Bruce~A. Reed, and William~T. Trotter.
\newblock Threshold tolerance graphs.
\newblock {\em J. Graph Theory}, 12(3):343--362, 1988.

\bibitem{Olariu91}
Stephan Olariu.
\newblock An optimal greedy heuristic to color interval graphs.
\newblock {\em Inf. Process. Lett.}, 37(1):21--25, 1991.

\bibitem{PachT19}
J{\'{a}}nos Pach and Istv{\'{a}}n Tomon.
\newblock Ordered graphs and large bi-cliques in intersection graphs of curves.
\newblock {\em Eur. J. Comb.}, 82, 2019.

\bibitem{PachT20}
J{\'{a}}nos Pach and Istv{\'{a}}n Tomon.
\newblock On the chromatic number of disjointness graphs of curves.
\newblock {\em J. Comb. Theory, Ser. {B}}, 144:167--190, 2020.

\bibitem{RamalingamR88}
G.~Ramalingam and C.P. Rangan.
\newblock A unified approach to domination problems on interval graphs.
\newblock {\em Information Processing Letters}, 27(5):271--274, 1988.

\bibitem{Seinsche74}
D.~Seinsche.
\newblock On a property of a clas of n-colorable graphs.
\newblock {\em J. of Combinatorial Theory B}, 16:191--193, 1974.

\bibitem{Sinden66}
Frank~W Sinden.
\newblock Topology of thin film rc circuits.
\newblock {\em Bell System Technical Journal}, 45(9):1639--1662, 1966.

\bibitem{SotoC15}
Mauricio Soto and Christopher~Thraves Caro.
\newblock p-{B}ox: {A} new graph model.
\newblock {\em Discrete Mathematics {\&} Theoretical Computer Science},
  17(1):169--186, 2015.

\bibitem{S93}
J.P. Spinrad.
\newblock Doubly lexical ordering of dense 0-1 matrices.
\newblock {\em Inf. Proc. Letters}, 45:229--235, 1993.

\bibitem{Unger88}
Walter Unger.
\newblock On the k-colouring of circle-graphs.
\newblock In {\em {STACS} 88, 5th Annual Symposium on Theoretical Aspects of
  Computer Science}, volume 294 of {\em Lecture Notes in Computer Science},
  pages 61--72. Springer, 1988.

\bibitem{V12}
Martin Vatshelle.
\newblock {\em New width parameters of graphs}.
\newblock PhD thesis, University of Bergen, Norway, 2012.

\bibitem{Wood06}
David~R. Wood.
\newblock Characterisations of intersection graphs by vertex orderings.
\newblock {\em Australas. {J} Comb.}, 34:261--268, 2006.

\end{thebibliography}
\bibliographystyle{plain}

\end{document}